\pdfoutput=1
\documentclass[conference, 10pt, letterpaper]{IEEEtran}
\usepackage{cite}
\ifCLASSINFOpdf
   \usepackage[pdftex]{graphicx}
   \graphicspath{{./}{./img/}}
   \DeclareGraphicsExtensions{.pdf,.jpeg,.jpg,.png,.eps,.ps}
\else
\fi
\usepackage[cmex10]{amsmath}
\usepackage{algorithmic}
\usepackage{array}

\usepackage{mdwmath}
\usepackage{mdwtab}

\usepackage{eqparbox}
\usepackage[font=footnotesize]{subfig}	%

\usepackage{stfloats}
\usepackage{url}
\usepackage[lined, ruled, linesnumbered, noend]{algorithm2e}
\usepackage[bookmarks=false]{hyperref}
\usepackage{latexsym}
\usepackage{amsfonts,amssymb}
\usepackage{color}
\let\labelindent\relax
\usepackage{enumitem}
\usepackage{pbox}
\usepackage{mathtools}
\usepackage{float}
\usepackage{pgffor}
\usepackage{caption}
\usepackage{amsthm}

\IEEEoverridecommandlockouts

\makeatletter
\def\thm@space@setup{%
  \thm@preskip=0.5em
  \thm@postskip=-0em%
}
\makeatother

\renewcommand{\thesection}{\arabic{section}}
\newtheorem{definition}{\bf Definition}	[section]
\newtheorem{theorem}{\bf Theorem}		[section]

\newtheorem{lemma}{\bf Lemma}		[section]

\newtheorem{remark}{\bf Remark}		[section]

\newcounter{appdx}
\newcommand{\kETAL}     {{\em et~al.}}		%
\newcommand{\red}[1]{{#1}}
\newcommand{\blue}[1]{{#1}}

\newcommand{\cmts}[1]{{#1}}

\newcommand{\myendbox}{\hfill $\Box$}

\begin{document}

\title{Survivable and Bandwidth-Guaranteed Embedding of Virtual Clusters in Cloud Data Centers \\(Extended Version)\thanks{
Yu, Xue and Zhang (\{ruozhouy, xue, xzhan229\}@asu.edu) are all with Arizona State University, Tempe, AZ 85287. 
Li (tolidan@tsinghua.edu.cn) is with Tsinghua University, Beijing, China. 
This research was supported in part by NSF grants 1457262 and 1421685, the National Key Basic Research Program of China (973 program) under Grant 2014CB347800, and the National Key Research and Development Program of China under Grant 2016YFB1000200. 
The information reported here does not reflect the position or the policy of the funding agencies.
This is the extended version of \cite{infocom17original} which has been accepted for publication in IEEE INFOCOM 2017.
}
\vspace{-0.4in}
}

\author{\IEEEauthorblockN{Ruozhou Yu, Guoliang Xue, Xiang Zhang, Dan Li}}

\maketitle
\thispagestyle{plain}
\pagestyle{plain}

\begin{abstract}
Cloud computing has emerged as a powerful and elastic platform for internet service hosting, yet it also draws concerns of the unpredictable performance of cloud-based services due to network congestion.
To offer predictable performance, the virtual cluster abstraction of cloud services has been proposed, which enables allocation and performance isolation regarding both computing resources and network bandwidth in a simplified virtual network model.
One issue arisen in virtual cluster allocation is the survivability of tenant services against physical failures.
Existing works have studied virtual cluster backup provisioning with fixed primary embeddings, but have not considered the impact of primary embeddings on backup resource consumption.
To address this issue, in this paper we study how to embed virtual clusters survivably in the cloud data center, by jointly optimizing primary and backup embeddings of the virtual clusters.
We formally define the survivable virtual cluster embedding problem. %
We then propose a novel algorithm, which computes the most resource-efficient embedding given a tenant request.
Since the optimal algorithm has high time complexity, we further propose a faster heuristic algorithm, which is several orders faster than the optimal solution, yet able to achieve similar performance.
Besides theoretical analysis, we evaluate our algorithms via extensive simulations.
\end{abstract}

\begin{IEEEkeywords}
Virtual cluster, survivability, bandwidth guarantee
\end{IEEEkeywords}

\IEEEpeerreviewmaketitle

\section{Introduction}
\label{sec:intro}

\noindent Cloud computing has emerged as a trending platform for hosting various services for the public.
Among different cloud computing platforms, the Infrastructure-as-a-Service (IaaS) clouds offer virtualized computing infrastructures to public tenants in order to host tenant services.
Enabled by advanced resource virtualization technologies, the clouds support intelligent resource sharing among multiple tenants, and can provision resources per tenant demand.

However, due to the massive migration of services to the cloud, there is increasing concern about the unpredictable performance of cloud-based services.
One major cause is the lack of network performance guarantee. %
All the tenants have to compete in the congested cloud network in an unorganized manner.
This has motivated recent efforts on cloud resource sharing with network bandwidth guarantee, for which a novel cloud service abstraction has been proposed, named virtual cluster (VC)~\cite{Ballani2011b}.
The VC abstraction allows each tenant to specify both the virtual machines (VMs) and per-VM bandwidth demand of its service. %
The cloud then realizes the request by allocating VMs on physical machines (PMs), as well as reserving sufficient bandwidth in the network to guarantee the bandwidth demand in the \emph{hose model}~\cite{Ballani2011b, Zhu2012a}.
The process of resource allocation for virtual clusters is called \emph{Virtual Cluster Embedding} (VCE).
Algorithms have been developed for VCE with various objectives and constraints~\cite{Ballani2011b, Zhu2012a, Xie2012, Yu2014a}.

One missing perspective in existing VCE solutions is the \emph{availability} of tenant services.
Due to the large-scale nature of cloud data centers, PM failures can happen frequently in the cloud~\cite{Yalagandula2004}.
When such failure happens, all services who have their VCs fully or partly embedded on the failed PMs will be affected, possibly receiving degraded service performance or even interruption of operation.
This not only impairs the tenants' interests, but also incurs additional cost to the cloud due to violation of service-level agreements.

To achieve the high availability goal of tenant services, one common practice is to enable service \emph{survivability}, by utilizing extra resources to help services recover quickly when actual failures happen.
A survivability mechanism can be either \emph{pro-active} or \emph{reactive}.
A pro-active mechanism provisions backup resources at the time of service provisioning, prior to the actual happening of failures.
Due to this, it can offer guaranteed recovery against a certain level of failures in the substrate, at the cost of underutilized resources when no failure is present.
On the contrary, a reactive mechanism only looks for backups as a reaction to actual failures.
While this means less reserved resources in the normal operation, a reactive mechanism may not always find a feasible recovery during the failure, and thus cannot guarantee the survivability of the service.

In this paper, we study how to efficiently provide pro-active protection for tenant services under the VC model.
In particular, we aim at embedding tenant VC requests survivably such that they can recover from any single-PM failure in the data center, meanwhile minimizing the total amount of resources reserved for each tenant.
\blue{We formally define survivable VC embedding as a joint resource optimization problem of both primary and backup embeddings of the VC.}
Following existing work~\cite{Ballani2011b, Zhu2012a}, we assume the data center has a tree structure, which abstracts many widely-adopted data center architectures.
We then propose an algorithm to optimally solve the embedding problem, within time bounded by a polynomial of the network size and the number of requested VMs (pseudo-polynomial time to input size).
The algorithm is based on the observation that the embedding decisions are independent for each subtree in the same level.
Since the optimal approach is time-consuming, we further propose a faster heuristic algorithm, whose performance is comparable to the optimal in practical settings.
We conduct both theoretical analysis and simulation-based performance evaluation, which have validated the effectiveness of our proposed algorithms.
Our main contributions are summarized as follows:
\begin{itemize}
\blue{\item To the best of our knowledge, we are the first to study the survivable and bandwidth-guaranteed VC embedding problem with joint primary and backup optimization.
\item We propose a pseudo-polynomial time algorithm that finds the most resource-efficient survivable VC embedding for each tenant request.}
\item We further propose a heuristic algorithm that reduces the time complexity of the optimal algorithm by several orders, yet has similar performance in the online scenario.
\item We use extensive experiments to evaluate the performance of our proposed algorithms.
\end{itemize}

The rest of this paper is organized as follows.
Section~\ref{sec:rw} presents the background and related work on VC embedding, as well as survivable cloud service provisioning.
Section~\ref{sec:model} describes the network service model, introduces our pro-active survivability mechanism, and formally defines the survivable VC embedding problem.
Section~\ref{sec:a} presents our optimal algorithm, and theoretical analysis for the proposed algorithm.
Section~\ref{sec:heu} presents our efficient heuristic algorithm and proves its feasibility.
Section~\ref{sec:eval} shows the evaluation results of our proposed algorithms, compared to a baseline algorithm.
Section~\ref{sec:conclusions} concludes this paper.

\section{Background and Related Work}
\label{sec:rw}

\subsection{Virtual Cluster Abstraction}

\noindent {Virtual cluster} (VC) is a newly proposed cloud service abstraction, which offers bandwidth guarantee over existing VM-based abstractions~\cite{Ballani2011b}.
In the VC model, the tenant submits its service request in terms of both the number of VMs and the per-VM bandwidth demand.
A tenant request, defined as a tuple $\langle N, B \rangle$, specifies a virtual topology where $N$ uniform VMs are connected to a central virtual switch, each via a virtual link with bandwidth of $B$, as shown in \cmts{Fig.~\ref{fig:vc}}.
To fulfill the request, the cloud should provision $N$ VMs in the substrate data center, with bandwidth guaranteed in the \emph{hose model} (to be detailed in Section~\ref{sec:model}).
In short words, hose model brings two major benefits: reduced model complexity (user specifies per-VM bandwidth instead of per-VM pair bandwidth as in the traditional \emph{pipe model}~\cite{Duffield1999}), and simple characterization of the minimum bandwidth requirement on each link~\cite{Ballani2011b}; interested readers are referred to~\cite{Ballani2011b} and~\cite{Duffield1999} for details.

\vspace{0.1in}
\begin{figure}[h!]
\centering
\includegraphics[width=0.25\textwidth]{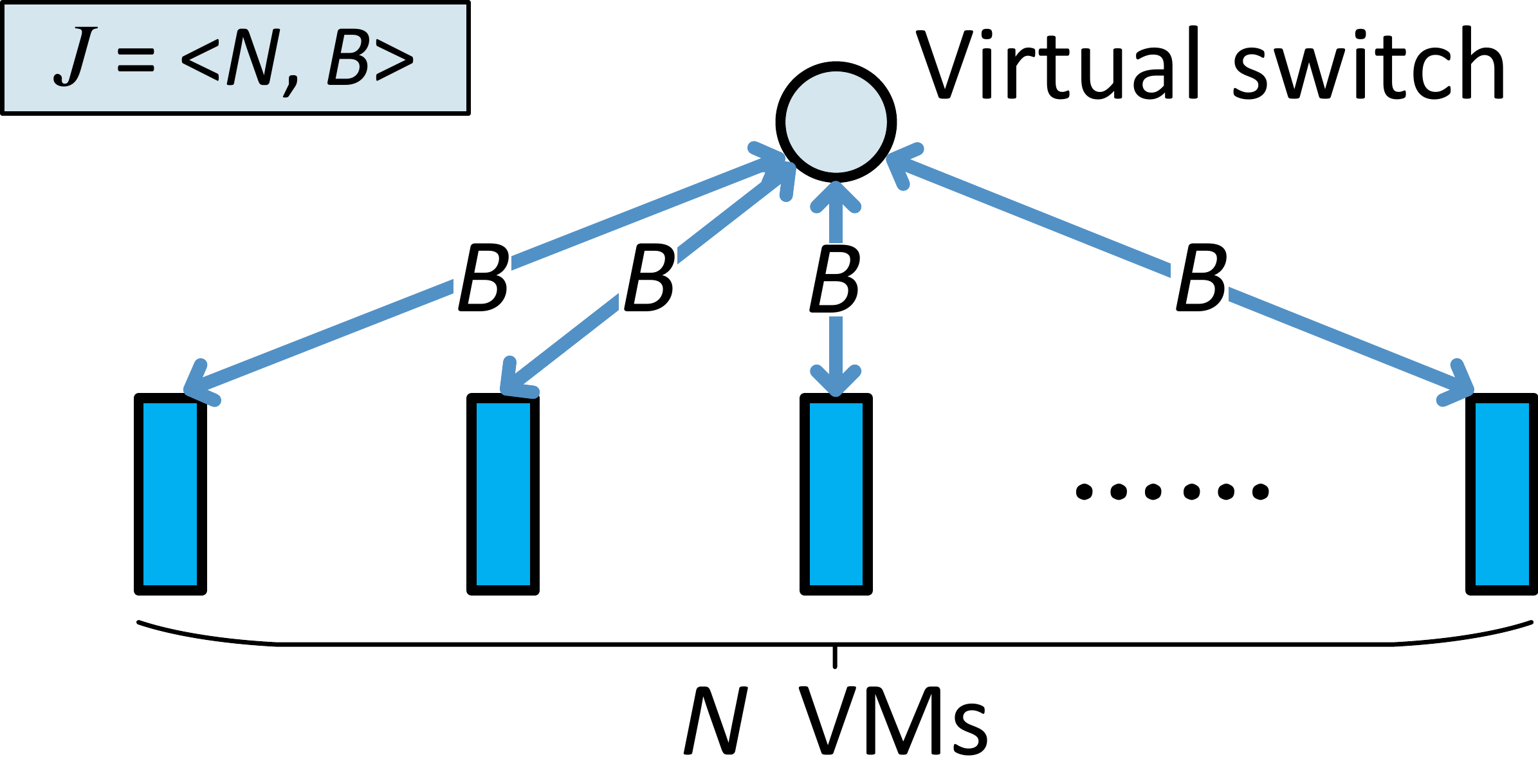}
\caption{Virtual cluster abstraction.}
\label{fig:vc}
\end{figure}

Ballani~\kETAL~\cite{Ballani2011b} first proposed the VC abstraction for cloud services with hose-model bandwidth guarantee.
They characterized the minimum bandwidth required on each link to satisfy the hose-model bandwidth guarantee, and developed a recursive heuristic for computing the VC embedding with minimum bandwidth consumption.
Based on it, Zhu~\kETAL~\cite{Zhu2012a} proposed an optimal dynamic programming algorithm to embed VC in the lowest subtree in tree-like data center topologies.
They also proposed a heuristic algorithm for VCs with heterogeneous bandwidth demands of their VMs.
TIVC~\cite{Xie2012} extends Oktopus with a time-related VC model that takes into consideration the dynamic traffic patterns of cloud services.
SVC~\cite{Yu2014a} also extends Oktopus,
and considers the statistical distribution of bandwidth demands between VMs. 
It proposes another dynamic programming algorithm to tackle the uncertain bandwidth demands.
DCloud~\cite{Li2015c} incorporates deadline constraints into the VC abstraction.
Instead of guaranteeing per-VM bandwidth, it guarantees that each accepted job will finish execution within its specified deadline.
\red{In a recent work, Rost~\kETAL~\cite{Rost2015} proposed that the VC embedding problem could be solved in polynomial time.
Yet, their model does not capture the minimum bandwidth required on each link to satisfy VM bandwidth requirements, which was characterized originally in~\cite{Ballani2011b}; as a result, their solution may over-provision bandwidth for VCs.
Recently, elastic bandwidth guarantee has drawn attention~\cite{Yu, Costa2016}.
Yu~\kETAL~\cite{Yu} proposed dynamic programming algorithms for dynamically scaling VCs, optimizing virtual cluster locality and VM migration cost.
Fuerst~\kETAL~\cite{Costa2016} also studied VC scaling minimizing migrations and bandwidth; their approach relies on the concept of center-of-gravity, which is determined by the location of the central switch.

None of the above has considered survivability in VC embedding.
Existing survivability mechanisms, such as those shown in the next subsection, do not lead to satisfactory solutions when directly applied to VC embedding, due to their lack of consideration for bandwidth requirement and/or lack of performance guarantee.
This paper focuses on deriving theoretically guaranteed solutions for the survivable VC embedding problem, as well as promising (low-complexity) heuristics.

Other problems similar to VC embedding include bandwidth-guaranteed VM embedding~\cite{Breitgand2012} and virtual network/infrastructure embedding~\cite{Even2013}.
The former problem is topology-agnostic, and only considers bandwidth on edge links; the latter considers a more general model where the virtual topology can be arbitrary graphs, hence it commonly suffers from high model complexity (to be detailed below).
}

\subsection{Survivable Virtual Cloud Service}

\noindent 
Providing survivability guarantee for VCs has been studied by Alameddine~\kETAL~\cite{Alameddine2016, Alameddine2016b}. 
Given a fixed primary embedding of a VC, they proposed a heuristic solution to ensure 100\% survivability for the VC with minimum backup VMs and bandwidth~\cite{Alameddine2016}. 
They further considered inter-VC bandwidth sharing to reduce backup bandwidth in\cite{Alameddine2016b}. 
However, their solutions did not consider the impact of the primary embedding to backup resource consumption. 
In this paper, we propose to jointly optimize both primary and backup resources of a VC, which can result in reduced backup resource consumption. 
Also, we propose an optimal solution, rather than heuristic solutions in\cite{Alameddine2016, Alameddine2016b}.

Beyond the VC abstraction, many have studied offering survivable virtual cloud services under various computing and network models~\cite{Nagarajan2007, Machida2010, Bin2011, Yeow2010, Yu2011a, Xu2013, Zhang2014a, Bodik2012}.
A first line of research focuses on providing survivable VM hosting in the cloud.
Nagarajan~\kETAL~\cite{Nagarajan2007} proposed the first pro-active VM protection method, which leverages the live migration capability of the Xen hypervisor to protect VMs from detected failures.
Based on this, Machida~\kETAL~\cite{Machida2010} studied redundant VM placement to protect a service from $k$ PM failures.
Bin~\kETAL~\cite{Bin2011} also studied VM placement for $k$-survivability, and proposed a shadow-based solution for VMs with heterogeneous resource demands.
The above papers do not offer bandwidth guarantee for VMs.
Our work utilizes a similar survivability mechanism as in~\cite{Bin2011}, where a predicted physical failure will trigger migration of the affected VMs to its backup location.
Yet we consider bandwidth guarantee in addition to VM placement, which complicates the problem and differentiates our work from the above.

Along another line, many solutions focus on survivable service hosting using the virtual infrastructure (VI) abstraction~\cite{Yeow2010, Yu2011a, Xu2013, Zhang2014a}.
A VI is a general graph, where each node or link may have a different resource demand, and an embedding is defined as two mappings: virtual node (VM) mapping and virtual link mapping; \red{the \emph{pipe model} is used in the VI abstraction instead of the hose model as in the VC abstraction.}
For example, Yeow~\kETAL~\cite{Yeow2010}, Yu~\kETAL~\cite{Yu2011a} and Xu~\kETAL~\cite{Xu2013} investigated survivable VI embedding through redundancy.
They formulated the problem with various objectives and constraints, and designed heuristic algorithms.
From a different angle, Zhang~\kETAL~\cite{Zhang2014a} proposed heuristic algorithms to embed VIs based on the availability statistics of the physical components.
The VI abstraction is more general than the VC abstraction, however, it is both hard to analyze theoretically and difficult to implement in large-scale networks due to its intrinsic model complexity~\cite{Ballani2011b}.
\red{An even more general model was proposed by Guo~\kETAL~\cite{Guo2010}, where a bandwidth requirement matrix is used to describe the bandwidth demand between each and every pair of virtual nodes; it suffers from even higher model complexity than the VI abstraction.}
CloudMirror~\cite{Lee} proposes a tenant application graph model for bandwidth guarantee, and discusses a heuristic opportunistic solution for high-availability that balances between bandwidth saving and availability.
Bodik~\kETAL~\cite{Bodik2012} studied general service survivability in bandwidth-constrained data centers.
Based on the service characteristics of Bing, they proposed an optimization framework and several heuristics to maximize fault-tolerance meanwhile minimizing bandwidth under the pipe model (per-VM pair bandwidth demand).
Survivability has been studied extensively in conventional communication networks and optical networks~\cite{Xue2003a, Xu2004a, Networks2014}.
Existing work focuses on providing connectivity guarantee against network link and switch failures.
The problem studied in this paper focuses on protecting tenant services from PM failures, which are different from link and switch failures.

\section{Network Model and Problem Statement}
\label{sec:model}

\noindent We study service provisioning in an IaaS cloud environment, where the cloud offers services in the form of inter-connected VMs.
To request a service, the tenant submits its request in terms of both VMs and network bandwidth.
A cloud hypervisor processes requests in an online manner.
For each request, the hypervisor first attempts to allocate enough resources in the data center.
If the allocation succeeds, it then reserves the allocated resources and provisions the VC for the tenant.
The VC will exclusively use all reserved resources until the end of its usage period, when the hypervisor will then revoke all allocated resources of the VC.
If the allocation fails due to lack of resources, the hypervisor rejects the request.

\subsection{Network Service Model}

\noindent Formally, each tenant request is defined as $\mathcal{J} = \langle N, B \rangle$, where $N$ is the number of requested VMs, and $B \ge 0$ is the per-VM bandwidth demand.

Following existing work~\cite{Ballani2011b, Zhu2012a}, we assume that the data center has a tree-structure topology.
In fact, many commonly used data center architectures have tree-like structures (FatTree~\cite{Al-Fares2008}, VL2~\cite{Greenberg2009}, etc.; \red{see Section~\ref{sec:disc}}), where our proposed algorithms can be adopted with simple abstraction of the substrate.
The substrate data center is defined as an undirected {tree} $T = (V, L)$, where $V$ is the set of nodes, and $L$ is the set of physical links.
The node set is further partitioned into two subsets $V = H \cup S$, where $H$ is the set of PMs that host VMs, and $S$ is the set of {abstract} switches which perform networking functions.
{Note that each abstract switch can represent a group of physical switches in the data center.}
Each PM is a leaf node, while each switch is an intermediate node in the topology.

{Without loss of generality, the substrate can be viewed as a rooted tree, and we pick a specific node $r \in S$ as its root, which generally represents all core switches.}
For each node $v$, we use $T_v$ to denote the subtree rooted at $v$.
We use $l_v \in L$ to denote the out-bound link of $T_v$, \emph{i.e.}, the link adjacent to node $v$ and on the shortest path from $v$ to global root $r$.
We use $d_v \ge 0$ to denote the number of children of $v$ in the tree.

For each PM $h \in H$, we define $c_h$ as the number of available VM slots on $h$.
For each node $v \in V$, we define $b_v$ as the available bandwidth on its out-bound link $l_v$.

\subsection{Virtual Cluster Embedding}

\noindent To fulfill a request $\mathcal{J} = \langle N, B \rangle$, the cloud needs to allocate $N$ VMs with bandwidth guarantee in \emph{the hose model}~\cite{Duffield1999}.
Given subtree $T_v$, let $n_v$ be the number of VMs allocated in $T_v$, then the minimum bandwidth required on link $l_v$ is given by
\begin{equation}
\label{eq:b}
\mathcal{B}(l_v) = \min \{ n_v, N - n_v \} \cdot B
\end{equation}
\emph{i.e.}, the bandwidth demand of VMs either inside $T_v$ or outside $T_v$, whichever is smaller.
In other words, given link bandwidth $b_v$, the number of VMs allocated within $T_v$, $n_v$, must satisfy %
\begin{equation}
n_v \in [0, {b_v \over B}] \cup  [N - {b_v \over B}, N]
\end{equation}
For simplicity of illustration, we define $\Lambda_v = \left( [0, {b_v \over B}] \cup  [N - {b_v \over B}, N] \right) \cap [0, N]$ as the \emph{feasible range of VMs w.r.t. the bandwidth of node $v$}, and $\overline \Lambda_v$ as its complement set over the universe $[0, N]$.
We also define $\lambda_v = \lceil N - {b_v \over B} \rceil$ as the lower bound of the upper feasible range, if $N - {b_v \over B} > N/2$.

\begin{definition}
Given substrate $T = (V, L)$ and request $\mathcal{J} = \langle N, B \rangle$, a {\bf Virtual Cluster Embedding (VCE)} is defined by a VM allocation function, 
$$\mathcal{C} : H \mapsto \mathbb{Z}^*$$
denoting the number of VMs allocated on each host, which satisfies the following properties: \\
\indent 1) $\mathcal{C}(h) \le c_h$ for any $h \in H$, \\
\indent 2) $\mathcal{B}(l_v) = \min \{ n_v, N - n_v \} \cdot B \le b_v$ for any $l_v \in L$, where $n_v = \sum\nolimits_{h \in T_v \cup H} { \mathcal{C}(h) }$, and \\
\indent 3) $\sum\nolimits_{h \in H} { \mathcal{C}(h) } = N$.
\myendbox
\end{definition}
Note that bandwidth allocation $\mathcal{B}$ is implicitly defined, as it can be computed based on VM allocation $\mathcal{C}$ as in Eq.~\eqref{eq:b}.

\begin{figure*}[tp!]
\centering
\subfloat[Legend]{
\includegraphics[width=0.12\textwidth]{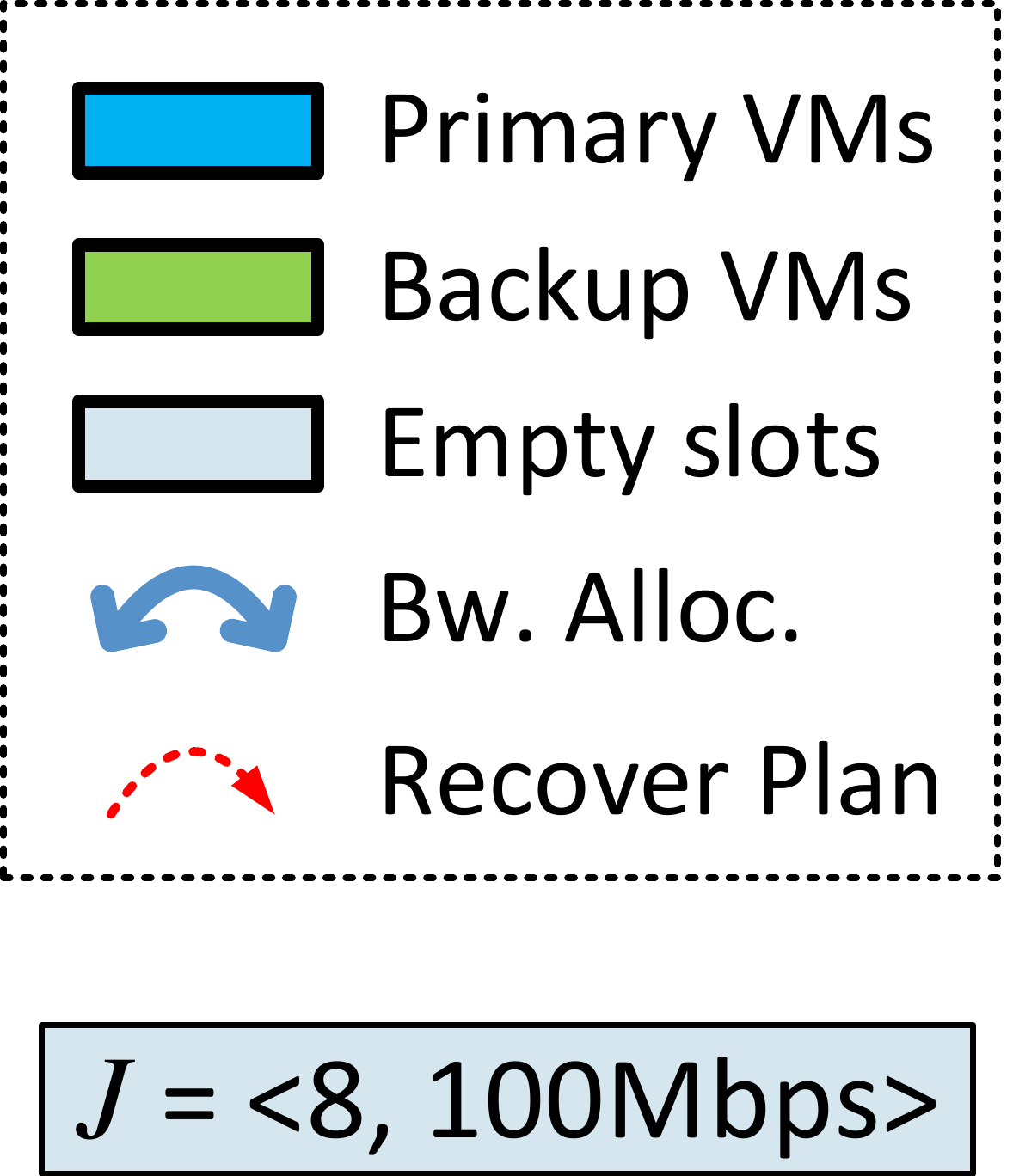}
\label{fig:legend}}
\hfil
\subfloat[VCE of $\mathcal{J}$ with $8$ VMs]{
\includegraphics[width=0.275\textwidth]{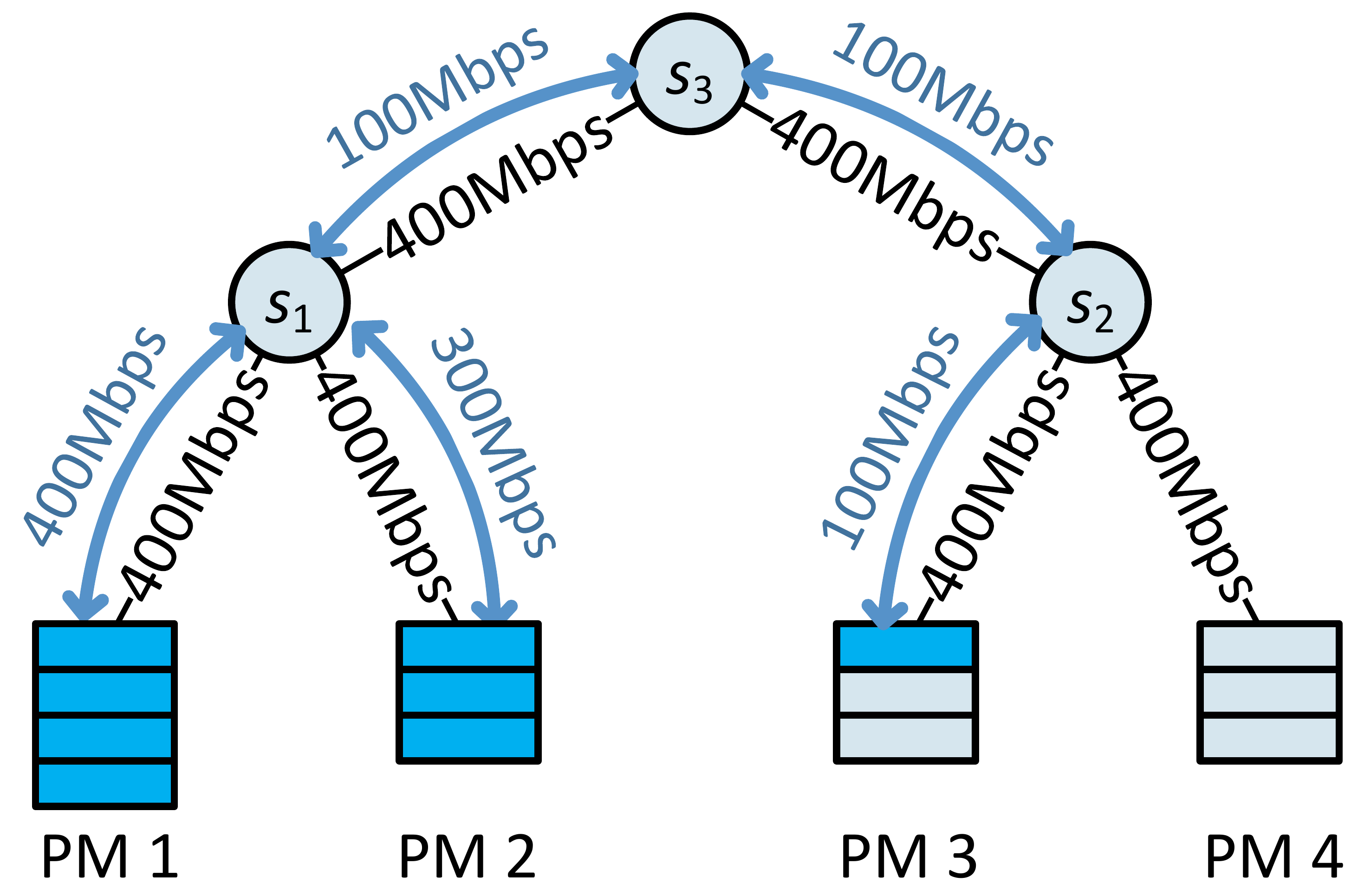}
\label{fig:vce}}
\hfil
\subfloat[SVCE of $\mathcal{J}$ with $12$ VMs in total]{
\includegraphics[width=0.275\textwidth]{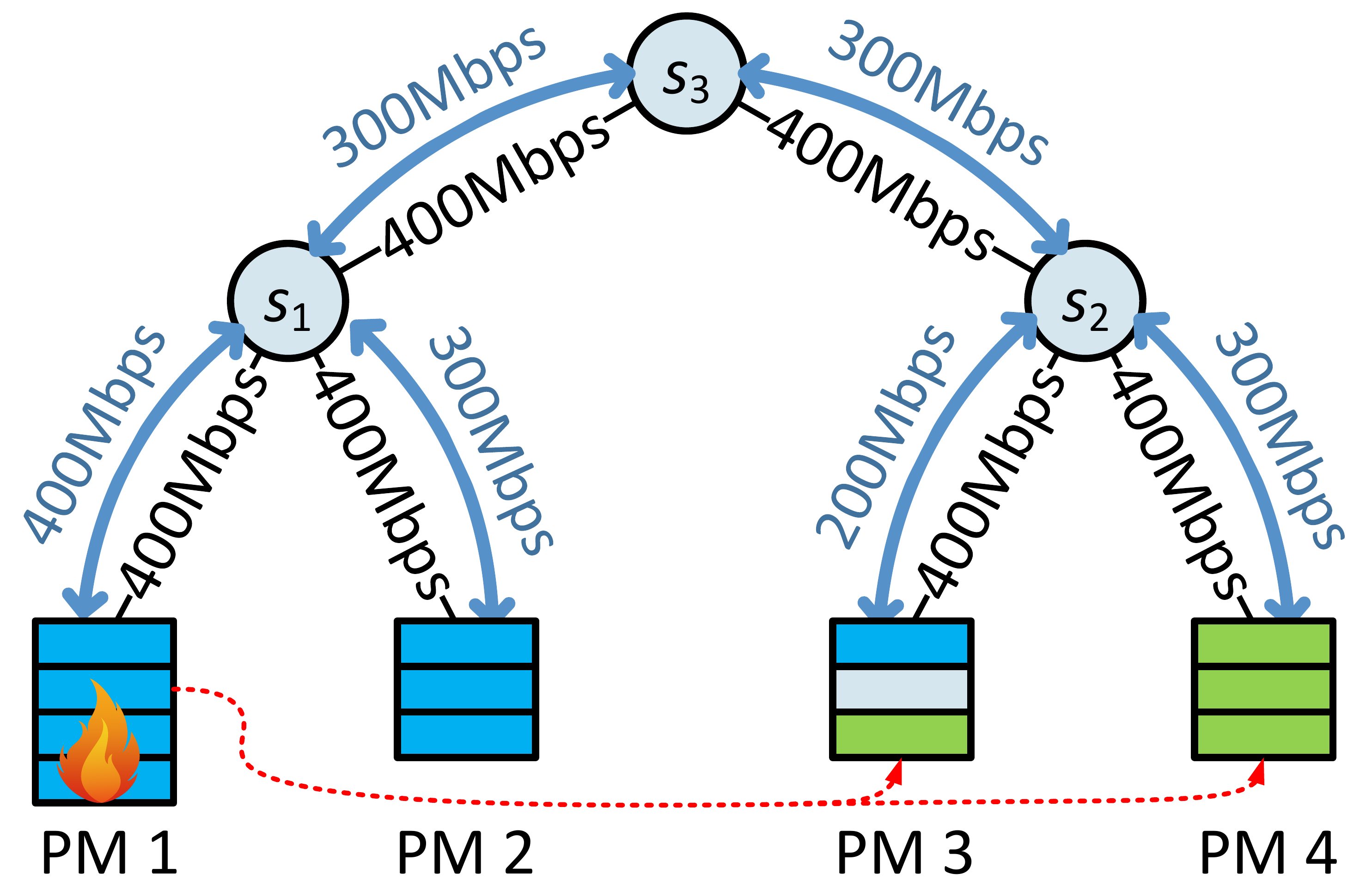}
\label{fig:svca}}
\subfloat[SVCE of $\mathcal{J}$ with $11$ VMs in total (optimal)]{
\includegraphics[width=0.275\textwidth]{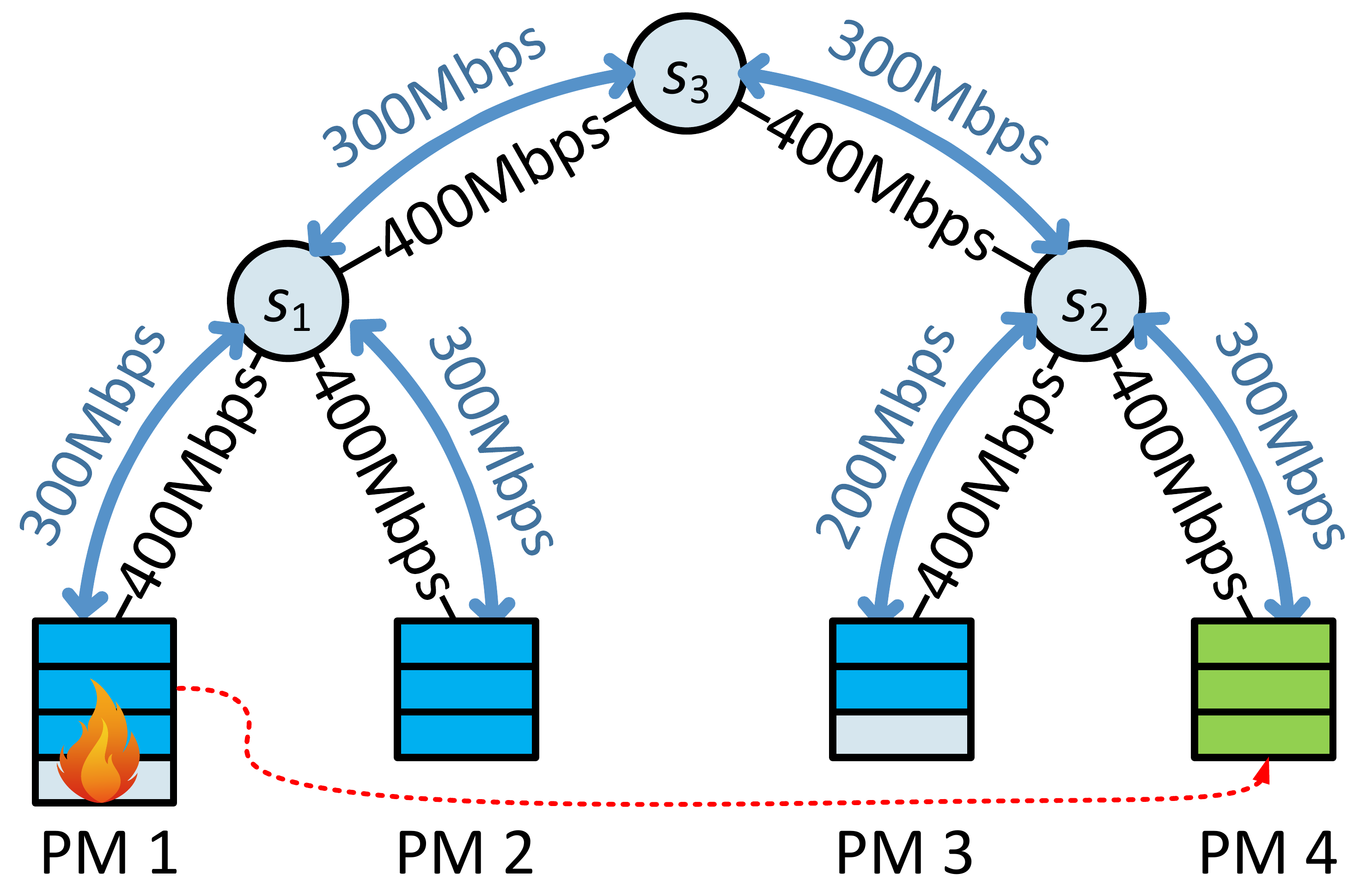}
\label{fig:svcao}}
\caption{VCE and SVCE comparison of $\mathcal{J} = \langle 8, 100 \text{Mbps} \rangle$.}
\label{fig:vcea}
\end{figure*}

\cmts{Fig.~\ref{fig:vcea}\subref{fig:vce} shows an embedding of the tenant VC request $\mathcal{J} = \langle 8, 100 \text{Mbps} \rangle$ on a $3$-level $2$-ary tree topology rooted at node $s_3$, where PM $1$ has $4$ VM slots and PMs $2$--$4$ each has $3$ VM slots, and each link has bandwidth of $400$Mbps.
$8$ VMs are allocated on PMs $1$, $2$ and $3$ to fulfill the request.
Note that link bandwidth constraints are satisfied based on Eq.~\eqref{eq:b}, as shown.
For example, although the subtree rooted at switch $s_1$ contains $7$ working VMs, only $\min \{ 7, 8-7 \} \cdot 100 = 100$Mbps bandwidth is required on its out-bound link as shown.}

If a tenant request is accepted, it then exclusively uses all its allocated resources, including both VM slots $\mathcal{C}(h)$ for $\forall h \in H$ and bandwidth $\mathcal{B}(l_v)$ for $\forall v \in V$.
This ensures guaranteed resources to the tenant service, leading to predictable service performance.
Finding VCE has been addressed in~\cite{Zhu2012a, Xie2012}.

\subsection{Survivability Mechanism}
\label{sec:surv}

\noindent Existing work for VCE does not consider service availability against physical failures.
For example, when a PM fails, all services with VMs hosted on it will be interrupted.
Moreover, due to lack of pre-provisioned backup resources, the cloud may not be able to recover the affected services in a short period of time.
This will lead to violated service-level agreements, and further economic losses to both the tenants and the cloud.

We use a {pro-active} survivability mechanism to improve service availability.
The idea is to pre-reserve dedicated backup resources for each service, and pre-compute the recovery plan against any possible failure scenario, during the initial embedding process.
During the life cycle of the service, a predicted physical failure will trigger the pre-determined automatic failover process, which will migrate the affected VMs to their backups.%
This way, the interruption period of the service is minimized.
{Note that while failures are frequent in the cloud, simultaneous PM failures are relatively rare~\cite{Yalagandula2004}.
Hence we only focus on the single-PM failure scenario, where each failure is defined by the failed PM alone: $F \in H$.
Link and switch failures are not considered, as modern data centers typically have rich path diversity between any pair of PMs~\cite{Al-Fares2008, Greenberg2009, Lee}, which can effectively protect over these failures.
}

To realize this mechanism, the key point is to reserve sufficient backup resources during the initial embedding process.
Specifically, the cloud needs to reserve both backup VM slots and backup bandwidth for the service.
{To characterize the total VM and bandwidth consumption, and the survivability guarantee of a VC, we define the following concept:}
{
\begin{definition}
Given substrate $T = (V, L)$ and request $\mathcal{J} = \langle N, B \rangle$, a {\bf Survivable Virtual Cluster Embedding (SVCE)} is defined by a tuple of allocation functions $(\mathcal{C}_s, \mathcal{B}_s)$, with 
$$\mathcal{C}_s : H \mapsto \mathbb{Z}^*$$ 
denoting the total number of VMs allocated on each PM, and 
$$\mathcal{B}_s : L \mapsto \mathbb{R}^*$$
denoting the total bandwidth allocated on each link, such that during any single-PM failure $F \in H$, there still exists a \emph{VCE} of $\mathcal{J}$, \cmts{in the auxiliary topology $T^{\mathcal{J}, F} = (V \setminus \{ F \}, L \setminus \{ l_F \})$ with resources on nodes and links defined as the remaining allocated resources, \emph{i.e.}, $c^{\mathcal{J}, F}_h = \mathcal{C}_s(h)$ for $\forall h \in H \setminus \{ F \}$ and $b^{\mathcal{J}, F}_v = \mathcal{B}_s(l_v)$ for $\forall v \in V \setminus \{ F \}$}.
\myendbox
\end{definition}
}

The above definition does not explicitly require a \cmts{VCE} in the normal operation (when no failure happens).
Such requirement is implicit, because after allocation, the \cmts{VCE} for any failure scenario can be used in the normal operation.
We call the \cmts{VCE} used in the normal operation as the \emph{primary working set} (PWS), and the \cmts{VCE} used in failure $F$ as the \emph{recovery working set} (RWS) regarding $F$.
We also use \emph{working VMs} to denote the set of VMs that are used (active) in a specific scenario, compared to the set of \emph{backup VMs} that remain inactive.
Note that given the SVCE, both working sets can be easily computed using existing VCE algorithms~\cite{Zhu2012a, Xie2012}.
The cloud pre-computes these VCEs in all possible scenarios, hence when failure happens, the recovery process can quickly find the backup resources needed for each affected VM.

{Fig.~\ref{fig:vcea}\subref{fig:svca} shows an SVCE of the tenant request $\mathcal{J} = \langle 8, 100 \text{Mbps} \rangle$.}
Compared to the VCE which allocates exactly $N=8$ VMs, in total $12$ VMs are provisioned in the SVCE.
During the failure of any PM, 1) the number of remaining VMs is always no less than $N=8$, and 2) a VCE exists under the hose model (with no more than $4$, $2$ and $3$ VMs on one side of the link PM $1$--$s_1$, the link PM $3$--$s_2$ and any other link respectively).
Hence the given SVCE can always recover the requested VC during any failure.
Note that we can assign the RWS during arbitrary failure as the PWS.
In this example, the dark blue VM slots on PM $1$ to $3$ are assigned as the PWS, while the green VM slots are backups.

\subsection{Resource Optimization}

\noindent {The problem we study is to find the SVCE that uses minimum resources in the cloud.
Moreover, we are interested in finding the SVCE that occupies the minimum number of VM slots, in order to accommodate as many future requests as possible.
Formally, we study the following optimization problem:}
\begin{definition}
Given substrate $T = (V, L)$ and request $\mathcal{J} = \langle N, B \rangle$, the {\bf Survivable Virtual Cluster Embedding Problem (SVCEP)} is to find an SVCE of request $\mathcal{J}$ that consumes the minimum number of VM slots in the substrate $T$.
\myendbox
\end{definition}

\cmts{The necessity of resource optimization is illustrated in Figs.~\ref{fig:vcea}\subref{fig:svca} and~\ref{fig:vcea}\subref{fig:svcao}.
While Fig.~\ref{fig:vcea}\subref{fig:svca} indeed shows an SVCE of $\mathcal{J} = \langle 8, 100 \text{Mbps} \rangle$, it consumes $4$ backup VMs, due to that a single failure at PM $1$ will affect $4$ VMs.
On the contrary, Fig.~\ref{fig:vcea}\subref{fig:svcao} shows a different SVCE that consumes only $3$ backup VM slots, and is optimal regarding total VM slots consumption.
With less consumed resources, the data center can accept more tenant requests in the future.}
{Note that although we focus on minimizing VM consumption, our proposed algorithms can be extended to minimize bandwidth as well; \red{see Section~\ref{sec:disc}}.}
\red{In the next two sections, we will present our proposed algorithms for solving SVCE.}

\section{Optimal Algorithm}
\label{sec:a}

\subsection{Algorithm Description}
\label{sec:vm}

\noindent We start from designing an algorithm that solves SVCEP optimally.
The algorithm works in a bottom-up manner: 
starting from the leaf nodes up to the root, the algorithm progressively determines the minimum number of total VMs needed in the subtree rooted at each node, each time solving a generalized problem of SVCEP.

Formally, define the following generalization of SVCEP:
\begin{definition}
Given substrate $T = (V, L)$, request $\mathcal{J} = \langle N, B \rangle$, an arbitrary node $v \in V$, and two nonnegative integers $n_0 \in [0, N]$ and $n_1 \in [0, N]$, the generalized problem {\bf SVCEP-GP} seeks to find the minimum number of VMs needed in $T_v$, to ensure that $T_v$ can provide \emph{at least} $n_0$ working VMs when no failure happens in $T_v$, and \emph{at least} $n_1$ working VMs during arbitrary (single-PM) failure in $T_v$.
\myendbox
\end{definition}
\begin{remark}
Both $n_0$ and $n_1$ concern not only the VM slots that can be offered by each child subtree of node $v$, but also the node's out-bound bandwidth $b_v$.
In other words, there exists a feasible solution to SVCEP-GP with $v$, $n_0$ and $n_1$ if and only if there exist two integers $\tilde n_0 \ge n_0$ and $\tilde n_1 \ge n_1$, such that 
\begin{equation*}
\max \{ \min \{ \tilde n_0, N - \tilde n_0 \}, \min \{ \tilde n_1, N - \tilde n_1 \} \} \le b_v 
\end{equation*}
and all child subtrees of $v$ can jointly offer exactly $\tilde n_0$ and $\tilde n_1$ VMs in the normal and worst-case failure scenarios (failure resulting in minimum number of available VMs) respectively.
\end{remark}
\begin{remark}
Note that given $v$, $n_0$ and $n_1$, a feasible solution of SVCEP-GP does not guarantee that the subtree $T_v$ can provide \emph{exactly} $n_0$ VMs if no failure happens in $T_v$, or $n_1$ VMs if arbitrary failure happens in $T_v$.
It only requires that $T_v$ can offer \emph{at least} $n_0$ or $n_1$ VMs in either scenario respectively.
For example, a subtree with out-bound bandwidth of $200$Mbps can offer $6$ VMs for request $\mathcal{J} = \langle 8, 100 \text{Mbps} \rangle$ if its child subtrees can jointly offer $6$ VMs, but cannot offer exactly $5$ VMs due to lack of bandwidth in the hose model.
However, as we will prove in the next subsection, the optimal solution to SVCEP-GP with $v = r$ and $n_0 = n_1 = N$ yields an optimal solution to the original SVCEP, and vice versa.
\end{remark}

Utilizing the above subproblem structure, we propose the following dynamic programming (DP) algorithm to compute the optimal solution by solving a sequence of SVCEP-GP instances.
Define $N_v[n_0, n_1]$ as the optimal solution to SVCEP-GP on node $v$, with non-negative integers $n_0$ and $n_1$.
The values are computed for PMs and switches as follows:

\noindent{\bf PM computation:}
For leaf node $h \in H$, we have
\begin{equation}
\label{eq:hnv}
N_h[n_0, n_1] = 
\left\{ {\begin{array}{*{20}{l}}
	{n_0}		& {\text{if } n_1 = 0, n_0 \in [0, c_h] \cap \Lambda_h}\\
	{\lambda_h}	& {\text{if } n_1 = 0, n_0 \in \overline \Lambda_h, \lambda_h \le c_h }\\
	\infty 		& {\text{otherwise}}
\end{array}} \right.
\end{equation}

\emph{Explanation}: 
The number of working VMs that a PM can offer is bounded by three factors: the number of available slots $c_h$, the out-bound bandwidth $b_h$, and the requested VM number $N$.
In the normal operation, the minimum number of VMs to offer $n_0$ working VMs is equal to $n_0$, if $n_0$ is in the feasible range bounded by bandwidth $b_h$.
Recall that $\Lambda_v$ ($\overline \Lambda_v$) defines the feasible (infeasible) range of working VMs in $T_v$ w.r.t. bandwidth, and $\lambda_v$ is the lower bound of the upper range of $\Lambda_v$.
If $n_0 \in \overline \Lambda_h$, the PM cannot offer exactly $n_0$ VMs due to the bandwidth limit; however, if there are at least $\lambda_h$ available slots, the PM can offer $\lambda_h$ VMs, which guarantees \emph{at least} $n_0$ VMs in the PWS with the minimum total VMs.
Note that $n_1$ always equals $0$, as when this PM fails, no VM can be offered.
All other entries are $\infty$, meaning such instances are infeasible.

\noindent{\bf Switch computation:}
The computation for a switch node is more complicated, as there are exponential number of ways to write an integer value ($n_0$ or $n_1$) as the sum of $d_v$ integer values, where $d_v$ is the number of children of node $v$.
However, we observe that the allocation in each subtree is independent from the other subtrees.
Hence we employ another level of dynamic programming to aggregate results from child subtrees.

Define $N_v' [n_0, n_1, k]$ as the minimum number of total VMs in $T_v$, to ensure that $T_v$ can provide at least $n_0$ working VMs in the normal operation, and at least $n_1$ working VMs during arbitrary failure in $T_v$, \emph{using the first $k$ subtrees of $T_v$}, where $k \in \{ 0, \cdots, d_v \}$.
Note that $N_v' [n_0, n_1, k]$ does not consider the out-bound bandwidth $b_v$ of node $v$.
We first establish the relationship between $N_v$ and $N_v'$ as:
\begin{equation}
\label{eq:snv}
N_v[n_0, n_1] = 
\left\{ {\begin{array}{*{20}{l}}
	{N'_v[n_0, n_1, d_v]}				& {\text{if } n_0, n_1 \in \Lambda_h }\\
	{N'_v[\lambda_v, n_1, d_v]}		& {\text{if } n_0 \in \overline \Lambda_h, n_1 \in \Lambda_h }\\
	{N'_v[n_0, \lambda_v, d_v]}		& {\text{if } n_0 \in \Lambda_h, n_1 \in \overline \Lambda_h }\\
	{N'_v[\lambda_v, \lambda_v, d_v]}	& {\text{if } n_0, n_1 \in \overline \Lambda_h }\\
\end{array}} \right.
\end{equation}
for each switch node $v \in S$ and $n_0, n_1 \in \{ 0, \cdots, N \}$.

\emph{Explanation}: 
Based on their definitions, $N_v [n_0, n_1]$ and $N_v' [n_0, n_1, d_v]$ only differ in that the latter does not consider the out-bound bandwidth at node $v$.
Hence we apply the bandwidth constraints to obtain $N_v [n_0, n_1]$ from $N_v'$: if $b_v$ cannot support $n_0$ ($n_1$) working VMs in $T_v$, then we take the minimum value $\lambda_v$ that both can be supported by $b_v$ and is at least $n_0$ ($n_1$) as desired.
Note that both $N_v[n_0, n_1]$ and $N'_v[n_0, n_1, k]$ are non-decreasing in either $n_0$ or $n_1$ based on definition.
Hence the above defined $N_v[n_0, n_1]$ is optimal given the optimality of all its dependent $N_v'[n_0', n_1', d_v]$ values, for $n_0' \in \{ n_0, \lambda_v \}$ and $n_1' \in \{ n_1, \lambda_v \}$.

\noindent{\bf Inner DP:}
The value of $N_v'[n_0, n_1, k]$ is computed from $k=0$ to $d_v$.
The initial iteration where $k=0$ is computed as:
\begin{equation}
\label{eq:hnvp}
N'_v[n_0, n_1, 0] = 
\left\{ {\begin{array}{*{20}{l}}
	{0}	& {\text{if } n_0 = n_1 = 0}\\
	\infty & {\text{otherwise}}
\end{array}} \right.
\end{equation}
since only $0$ VMs can be offered using $0$ subtrees.
Based on this, each of the other iterations computes $N_v'[n_0, n_1, k]$ based on the values $N_v'[n_0', n_1', k-1]$ computed in the $(k-1)$-th iteration as well as the values $N_{u_k}[n_0'', n_1'']$ computed for the $k$-th child node $u_k$ of node $v$, as shown in Eq.~\eqref{eq:snvp}.
\begin{multline}
\label{eq:snvp}
N'_v[n_0, n_1, k] = {\min\limits_{\substack{n_0', n_0'', \\ n_1', n_1''}} \left\{ N'_v[n_0', n_1', k-1] + N_{u_k}[n_0'', n_1'']  \left|\,
\begin{array}{*{20}{l}}
{\ }	\\
{\ }
\end{array}
\right. \right.} 
\\ 
{\left.  
\begin{array}{*{20}{l}}
{n_0' + n_0'' \ { \ge }\  n_0 }	\\
{\min \{ n_0' + n_1'',n_0'' + n_1' \}  \ { \ge }\  {n_1}}		
\end{array}
\right\} }
\end{multline}

\emph{Explanation}: 
The computation of $N'_v[n_0, n_1, k]$ iterates over four variables: $n_0'$, $n_1'$, $n_0''$ and $n_1''$.
The rationale is that, either in the normal operation or during arbitrary failure, the required working VMs in the first $k$ subtrees of $T_v$ can be split into two parts: the VMs in the first $(k-1)$ subtrees, and the VMs in the $k$-th subtree.
Moreover, since we assume that only single PM failure can happen at any time, once the failure happens within one part, the other part is assured to work in the normal operation.
Hence we know how the required working VMs $n_0$ and $n_1$ should be split between the two parts.
Specifically, we have $n_0' + n_0'' \ge n_0$ in the normal operation;
the failure can happen in either the first or the second part, and the worst-case failure is the one that results in fewer working VMs: $\min \{ n_0' + n_1'', n_0'' + n_1' \} \ge n_1$.

\begin{algorithm}
\caption{Find optimal SVCE minimizing VMs}
\label{a:1}
\KwIn{Topology $T = (V, L)$, request $\mathcal{J} = \langle N, B \rangle$}
\KwOut{SVCE $(\mathcal{C}_s, \mathcal{B}_s)$}
\For{each node $v$ from bottom to top}{
	\eIf{$v$ is a leaf node} {
		\For{$n_0, n_1 \in \{ 0, \dots, N \}$} {
			Compute $N_v [n_0, n_1]$ as in Eq.~\eqref{eq:hnv}\;
		}
	} {
		Compute $N'_v [n_0, n_1, 0]$ as in Eq.~\eqref{eq:hnvp} for $\forall n_0, n_1$\;
		\For{$k=1$ to $d_v$} {
        			\For{$n_0, n_1 \in \{ 0, \dots, N \}$} {
        				Compute $N'_v [n_0, n_1, k]$ as in Eq.~\eqref{eq:snvp}\;
        			}
		}
		\For{$n_0, n_1 \in \{ 0, \dots, N \}$} {
			Compute $N_v [n_0, n_1]$ as in Eq.~\eqref{eq:snv}\;
		}
	}
}
\eIf{$N_r [N, N] = \infty$} {
	\Return{Infeasible.}
} {
	Backtrack the DP process to find allocations $\mathcal{C}_s(h)$ for $\forall h \in H$ and $\mathcal{B}_s(l_v)$ for $\forall v \in V$\;
}
\Return{SVCE $(\mathcal{C}_s, \mathcal{B}_s)$.}
\end{algorithm}

Algorithm~\ref{a:1} shows the whole procedure of the dynamic programming.
The algorithm first computes all the entries of $N_v[n_0, n_1]$ in a bottom-up manner.
The order of computation guarantees that during the computation of an entry in either $N_v$ or $N_v'$, all its depending entries have already been computed in previous iterations.
After computation, if the value $N_r[N, N]$ is feasible, the algorithm then backtracks the DP process to obtain the exact VM and bandwidth allocations in each subtree.
VM allocation can be determined by recording the path towards each entry in the table during DP, while bandwidth can be determined based on $n_0$ and $n_1$ due to Eq.~\eqref{eq:b}.

\noindent \subsection{Algorithm Analysis}

The following theorems deliver the optimality and the running time of our proposed algorithm.
\red{For clarity of presentation, we refer readers to the appendix for rigorous proofs of the theorems and lemmas in this subsection.}

\begin{lemma}
\label{l:1}
Given an allocation of VMs in any subtree $T_v$ for $\mathcal{J} = \langle N, B \rangle$, if the subtree can offer $n_v \in (N/2, N]$ working VMs in a scenario, then it can offer any number of working VMs less than or equal to $n_v^- = N - n_v$ in the same scenario without increasing bandwidth on any link.
\myendbox
\end{lemma}

\begin{lemma}
\label{l:2}
Given an allocation of VMs in any subtree $T_v$ for $\mathcal{J} = \langle N, B \rangle$, if the subtree can offer \emph{more than} $N$ working VMs in a scenario, then it can offer \emph{exactly} $N$ VMs in the same scenario, without increasing bandwidth on any link.
\myendbox
\end{lemma}

\begin{theorem}
\label{th:1}
Given an instance of SVCEP, Algorithm~\ref{a:1} returns the optimal solution if the instance is feasible, and returns ``Infeasible'' otherwise.
\myendbox
\end{theorem}

\begin{theorem}
\label{th:2}
The worst-case time complexity of Algorithm~\ref{a:1} is bounded by $O(|V| \cdot N^6)$, where $|V|$ is the network size and $N$ is the request size (the number of VMs requested).
\myendbox
\end{theorem}

Based on Theorem~\ref{th:1}, the solution is guaranteed to offer a feasible VCE of $\mathcal{J} = \langle N, B \rangle$ using the allocated resources when facing any single PM failure.
To find the RWS for each failure, one can apply existing VCE algorithms~\cite{Zhu2012a, Xie2012} in the auxiliary topology where VM slots and bandwidth are the same as allocated except for the failed PM.
As mentioned in Section~\ref{sec:surv}, the RWS of any failure can be used as the PWS.

\section{Efficient Heuristic}
\label{sec:heu}

\noindent The algorithm proposed in Section~\ref{sec:a} optimally solves SVCEP.
However, its worst-case time complexity can be as high as $\Theta(|V| \cdot N^6)$, which may be too expensive when a tenant asks for many VMs.
In this section, we propose an efficient heuristic algorithm that runs in \cmts{$O(N \cdot |V| \log |V|)$} time.

Before the algorithm, we first state the following lemma, \red{whose proof is also detailed in the appendix:}
\begin{lemma}
\label{l:heu}
Given substrate $T$, request $\mathcal{J} = \langle N, B \rangle$, and an integer $N' \in [1, N]$, a VCE of the augmented request $\mathcal{J}' = \langle N+N', B \rangle$ yields a feasible SVCE of $\mathcal{J}$ as long as each PM is allocated with no more than $N'$ VMs in the VCE.
\myendbox
\end{lemma}

Based on Lemma~\ref{l:heu}, we design a heuristic algorithm shown in Algorithm~\ref{a:2}, based on the algorithm in~\cite{Zhu2012a}.
\begin{algorithm}
\caption{Heuristic for finding SVCE}
\label{a:2}
\KwIn{Topology $T = (V, L)$, request $\mathcal{J} = \langle N, B \rangle$}
\KwOut{SVCE $(\mathcal{C}_s, \mathcal{B}_s)$}
\For{$N'=1$ to $N$} {
	$\mathcal{C}_s, \mathcal{B}_s \leftarrow FindFeasibleVCE(T, \langle N+N', B \rangle)$;
	
	\If{$(\mathcal{C}_s, \mathcal{B}_s)$ is a feasible solution} {
		\Return{SVCE $(\mathcal{C}_s, \mathcal{B}_s)$.}
	}
}
\Return{Infeasible.}
\end{algorithm}

The algorithm iterates from $N'=1$ to $N$, each time calling $FindFeasibleVCE$ to find a feasible VCE for the augmented request $\mathcal{J}' = \langle N+N', B \rangle$. %
If such feasible embedding is found at a specific value of $N'$, the total number of VMs provisioned is exactly $N+N'$, hence yielding a solution with minimum VMs under this algorithm.
The algorithm stops when no solution is found in $N$ iterations, because the number of backup VMs should not exceed the number of requested VMs.

The subroutine $FindFeasibleVCE$ uses the algorithm in~\cite{Zhu2012a} to find a VCE for $\mathcal{J}'$.
Minor modification is done to enforce the per-PM VM limit $N'$.
Due to page limit, we only briefly introduce the idea of the algorithm.
For each node, the algorithm uses a data structure called the \emph{Allocation Range} (AR) to record the number of VMs that can be allocated within its subtree.
The AR consists of $N+N'+1$ bits for request $\mathcal{J}'$, where the $(i+1)$-th bit is $1$ if the subtree can offer $i$ VMs, and $0$ otherwise.
Continuous $1$ bits are aggregated into sections defined by both end-points.
For example, a section $[5, 7]$ means that the subtree can offer $5$, $6$ or $7$ VMs.
The algorithm progressively computes the AR of every node, from leaves to root.
It then finds the lowest subtree that can offer $N+N'$ VMs, and makes allocation through backtracking.

\cmts{
The algorithm in~\cite{Zhu2012a} finds a VCE within $O(|V| \log |V|)$ time.
Algorithm~\ref{a:2} calls $FindFeasibleVCE$ for at most $N$ times, hence it has time complexity $O( N \cdot |V| \log |V| )$.
}

Algorithm~\ref{a:2} does not guarantee optimality.
In fact, we can construct simple examples for which it fails to find an SVCE, yet one with the optimal objective can be found by our optimal algorithm.
Due to page limit, we omit the examples here.

However, as shown in Section~\ref{sec:eval}, this heuristic algorithm has similar performance to the per-request optimal solution proposed in Section~\ref{sec:a} when working in the online manner, but is several orders more time-efficient.
Therefore it is practically important for providing fast response to tenants in the cloud.

\section{Performance Evaluation}
\label{sec:eval}

\subsection{Baseline Algorithm}

\noindent \emph{Shadow-based solution} (SBS) is a well-known failover provisioning solution for VM management~\cite{Bin2011}.
In SBS, each primary VM is protected by a dedicated backup VM (called \emph{shadow}).
Different primary VMs do not share any common backup VM.
To employ SBS for VCs, both VMs and bandwidth need to be shadowed.
We designed a heuristic bandwidth-aware algorithm for SBS as our baseline algorithm.
It works as follows: given a request $\mathcal{J} = \langle N, B \rangle$, the algorithm seeks to find one primary VCE, as well as one shadow VCE, on two disjoint sets of PMs respectively.
When making the primary VCE, the algorithm seeks to minimize the PMs used, using a modified algorithm as in~\cite{Yu}, therefore leaving more room for the shadow.
A request is accepted only when the network accommodates both the primary VCE and the shadow.

We compared our proposed algorithms (OPT for the optimal algorithm and HEU for the heuristic algorithm) to this baseline algorithm (SBS) to show how resources are conserved to serve more requests by our optimization algorithms.

\subsection{Evaluation Metrics}

\begin{itemize}
\item \emph{Acceptance ratio} is the number of fulfilled requests over total requests, which directly reflects an algorithm's capability in serving as many requests as possible.
\item \emph{Average VM consumption ratio} is defined as the average ratio of actual VM slot consumption over the requested VMs, namely $\mathcal{R}_{\mathcal{J}} = (\sum\nolimits_{h} {\mathcal{C}_s(h)}) / N$, for each request. Note that this only counts those requests accepted by all three algorithms, in order to make fair comparison.
\item \emph{Average running time} reflects how much time an algorithm spends in average to determine a solution (or rejection) of each incoming request.
\end{itemize}

\subsection{Experiment Settings}

\noindent \red{We developed a C++-based simulator to evaluate our proposed algorithms.}
The substrate was simulated as a $4$-layer $8$-ary tree, including the PMs.
Each PM has $5$ VM slots, and $8$ PMs are connected to a Top-of-Rack (ToR) switch each via a $1$Gbps link.
$8$ ToR switches are connected to one aggregation switch, and $8$ aggregation switches to the core, both via $10$Gbps links.

We conducted experiments in two scenarios: the static scenario and the dynamic scenario.
In the static scenario, we used the same network information and the same tenant request in each experiment; \red{hence no resource was reserved after the acceptance of a request.}
To simulate realistic network states, we randomly generated load on PMs and links.
Specifically, given a load factor $\alpha$, we randomly occupied a fraction of the VM slots on each PM and bandwidth on each link, according to a normal distribution with mean of $\alpha$ and standard deviation of $\min \{ \alpha, 1.0 - \alpha \}$.
We then randomly generated $1000$ tenant requests each requesting $15$ VMs and $200$Mbps per-VM bandwidth on average with a normal distribution, and tested each of them on the network with random load.

In the dynamic scenario, 
we generated randomly arriving tenant requests, and embedded them in the initially unoccupied network in an online manner.
In each experiment $1000$ tenant requests were generated, which arrive in a Poisson process with mean arrival interval of $15$ and mean lifetime of $2000$.
Each request asks for $15$ VMs and $300$Mbps per-VM bandwidth on average, generated with a normal distribution.
\red{Resources were reserved after the acceptance of a request, hence existing VCs in the system would have impact on the embedding of future incoming VCs.}
Each experiment was repeated for \cmts{20} times in the same setting, and the results were averaged over all runs.

\red{In both scenarios, we varied one system parameter in each series of experiments, while keeping other parameters as default.}
Experiments were run on a Ubuntu Linux PC with Quad-Core 3.4GHz CPU and 16GB memory.

\subsection{Evaluation Results}

\subsubsection{Static Experiments}

\begin{figure}[ht!]
\vspace{-0.5em}
\centering
\subfloat[Acceptance ratio]{\includegraphics[height=0.155\textwidth]{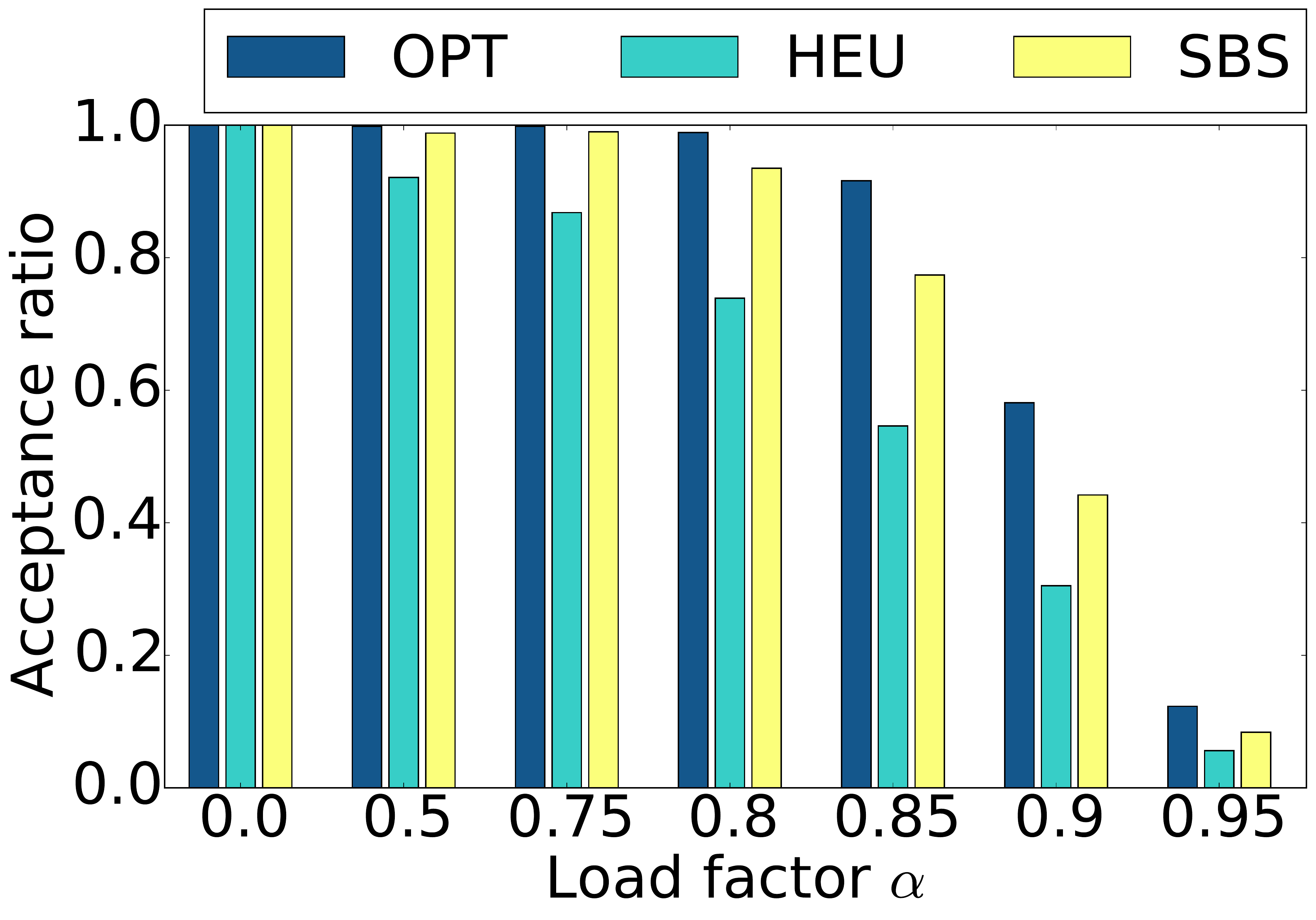}
\label{fig:load_acc}}
\hfil
\subfloat[Average VM consumption ratio]{\includegraphics[height=0.155\textwidth]{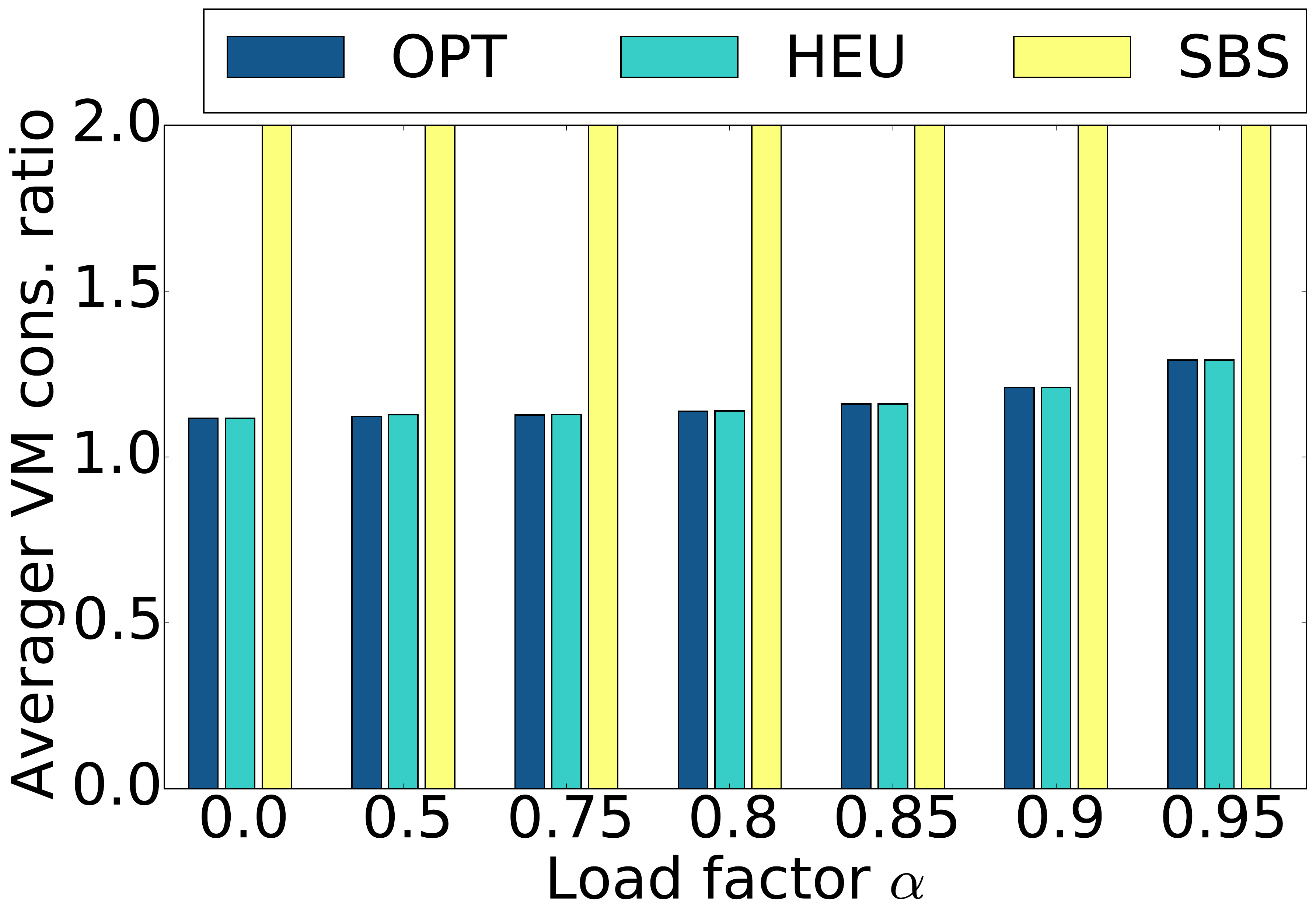}
\label{fig:load_vms}}
\caption{Static experiment with varying network load.}
\label{fig:load}
\end{figure}

Fig.~\ref{fig:load} shows the acceptance ratio and the VM consumption ratio with increasing network load (overall bandwidth consumption is similar to Fig.~\ref{fig:load}\subref{fig:load_vms} and is not shown due to page limit).
We observed that the OPT algorithm outperforms both HEU and SBS in terms of both number of requests accepted and the per-request VM consumption ratio, due to its optimality.
On the other hand, Fig.~\ref{fig:load}\subref{fig:load_acc} shows that the HEU algorithm performs less preferably than the SBS baseline in terms of acceptance ratio, when the network is loaded.
Further analysis reveals that HEU commonly requires more bandwidth from upper layer links, which are heavily congested by the random load, hence HEU's acceptance ratio is affected.
However, as can be observed in Fig.~\ref{fig:load}\subref{fig:load_vms}, HEU consumes much less VMs than SBS per accepted request.
Due to this, it is more likely for HEU to receive better performance when employed as an online scheduler, due to its capability in conserving cloud resources.
As will be shown next, HEU indeed outperforms SBS greatly in the dynamic experiments.
SBS always consumes $2\times$ the VMs requested, as it provisions an entire duplicate of the primary VMs.
Per-request VM consumption increases slightly with the increasing load, due to that it is harder to find a survivable embedding with few backup VMs when the network is short of bandwidth.
\subsubsection{Dynamic Experiments}
\begin{figure}[ht!]
\vspace{-0.5em}
\centering
\subfloat[Acceptance ratio]{\includegraphics[height=0.155\textwidth]{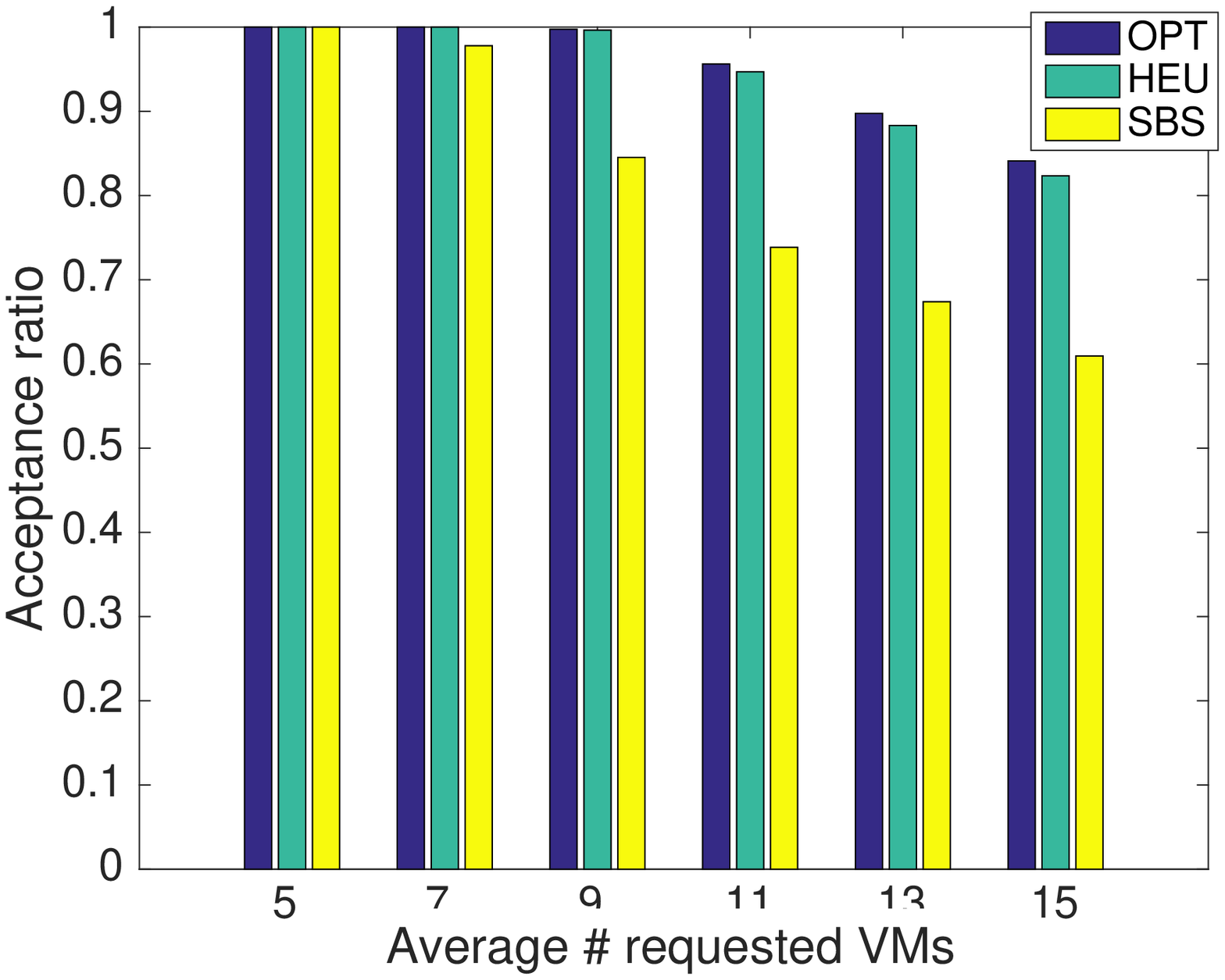}
\label{fig:vms_acc}}
\hfil
\subfloat[Average running time (log scale)]{\includegraphics[height=0.155\textwidth]{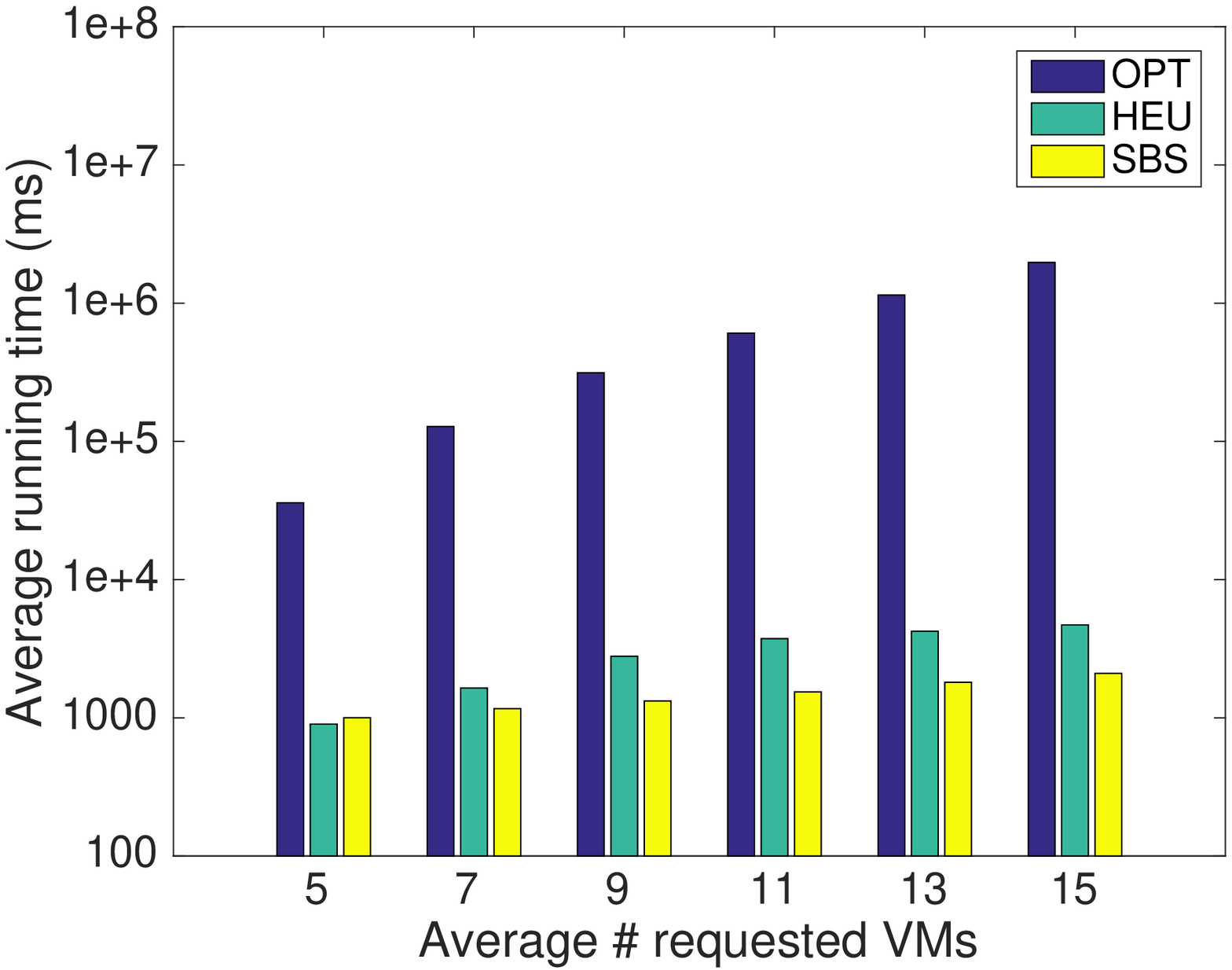}
\label{fig:vms_tm}}
\caption{Dynamic experiment with varying VMs per request}
\label{fig:vms}
\end{figure}

Fig.~\ref{fig:vms} shows the experiment results with varying average number of requested VMs per tenant request.
OPT obviously achieves the best acceptance ratio in all scenarios, while HEU's acceptance ratio is only slightly lower than OPT.
Both algorithms have much higher acceptance ratio compared to the SBS baseline, due to their capability to conserve VM (and bandwidth) resources per tenant request.
Meanwhile, HEU has much shorter running time than OPT in all cases due to its low time complexity, and is only a little worse than SBS in most scenarios.
As each tenant asking for more VMs, acceptance ratio drops while running time increases.
This is due to that the running time of all algorithms are related to the per-tenant request size $N$.
\begin{figure}[ht!]
\vspace{-0.3em}
\centering
\subfloat[Acceptance ratio]{\includegraphics[height=0.155\textwidth]{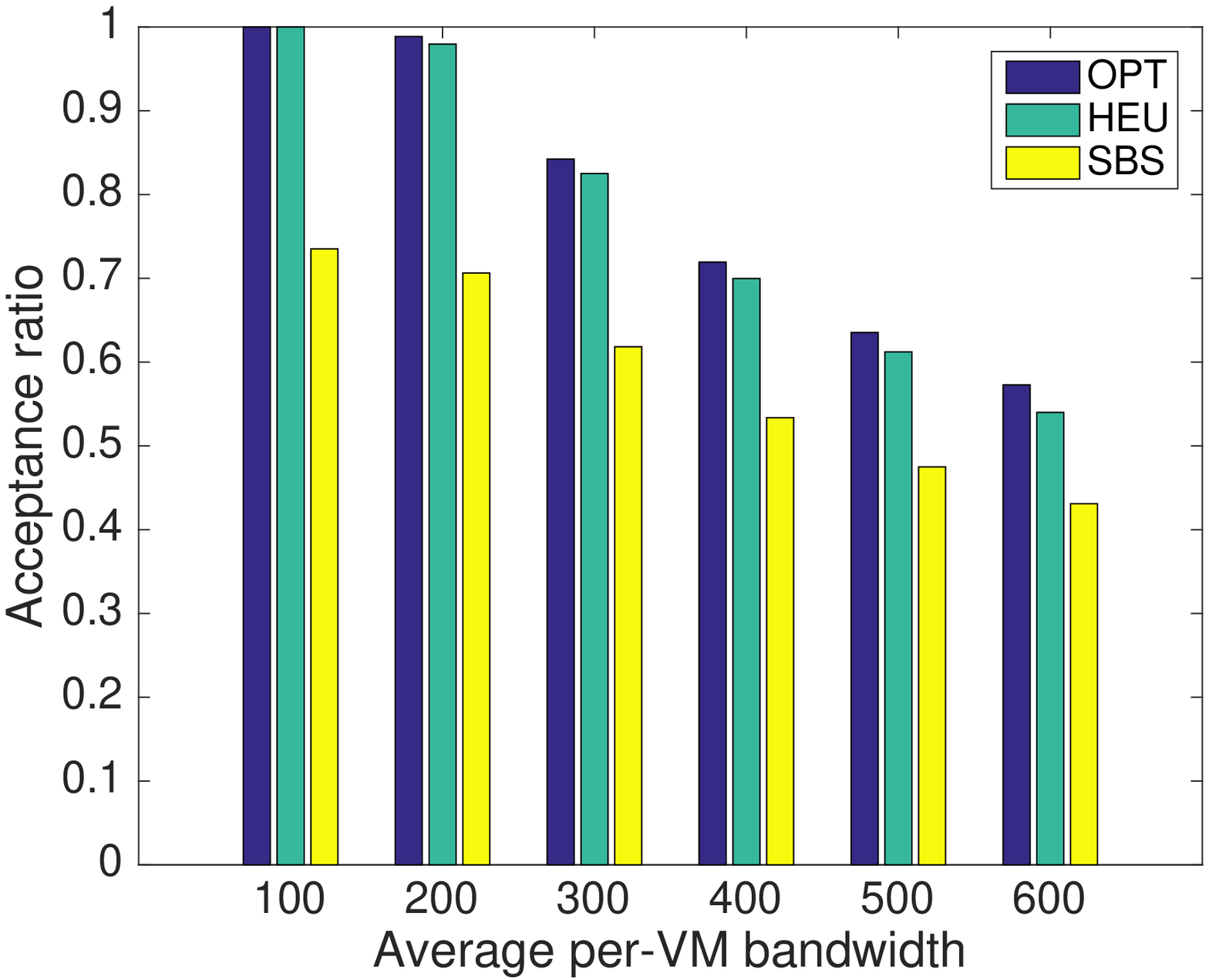}
\label{fig:dmd_acc}}
\hfil
\subfloat[Average running time (log scale)]{\includegraphics[height=0.155\textwidth]{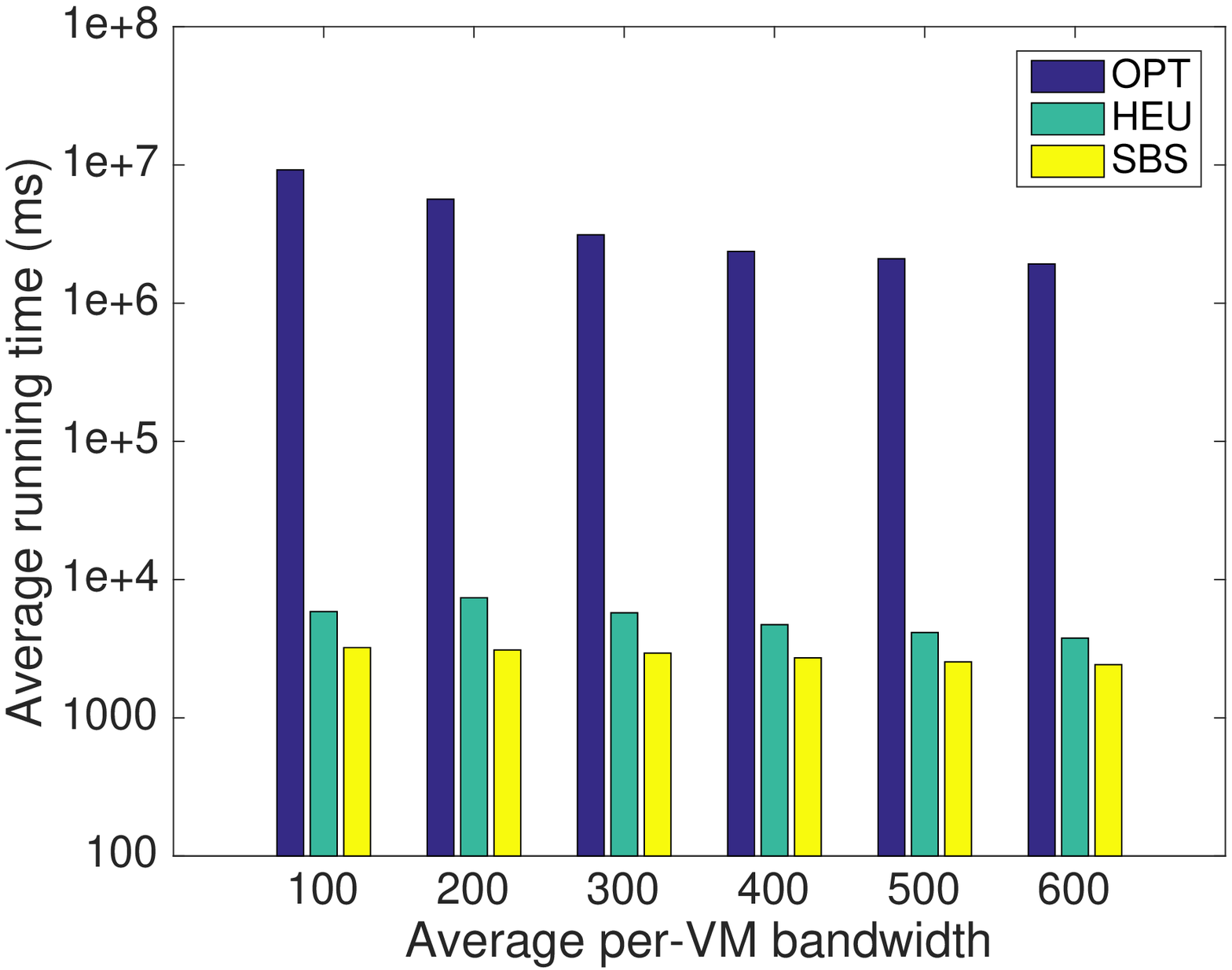}
\label{fig:dmd_tm}}
\caption{Dynamic experiment with varying per-VM bandwidth}
\label{fig:dmd}
\end{figure}

Fig.~\ref{fig:dmd} shows the experiment results with varying average per-VM bandwidth.
The acceptance ratio results show similar pattern as in Fig.~\ref{fig:vms}, where OPT performs the best, HEU performs slightly worse, and SBS performs much worse compared to the former two.
Acceptance ratio drops as per-VM bandwidth increases.
As for running time, clearly HEU and SBS are both much better than OPT due to their low complexity.
Unlike Fig.~\ref{fig:vms}, running time drops as per-VM bandwidth increases.
It has two reasons.
First, the worst-case time complexity of each algorithm is not related to the per-VM bandwidth.
Second, as per-VM bandwidth increases, the search spaces decrease due to more consumed network resources.

\begin{figure}[ht!]
\vspace{-0.5em}
\centering
\subfloat[Acceptance ratio]{\includegraphics[height=0.155\textwidth]{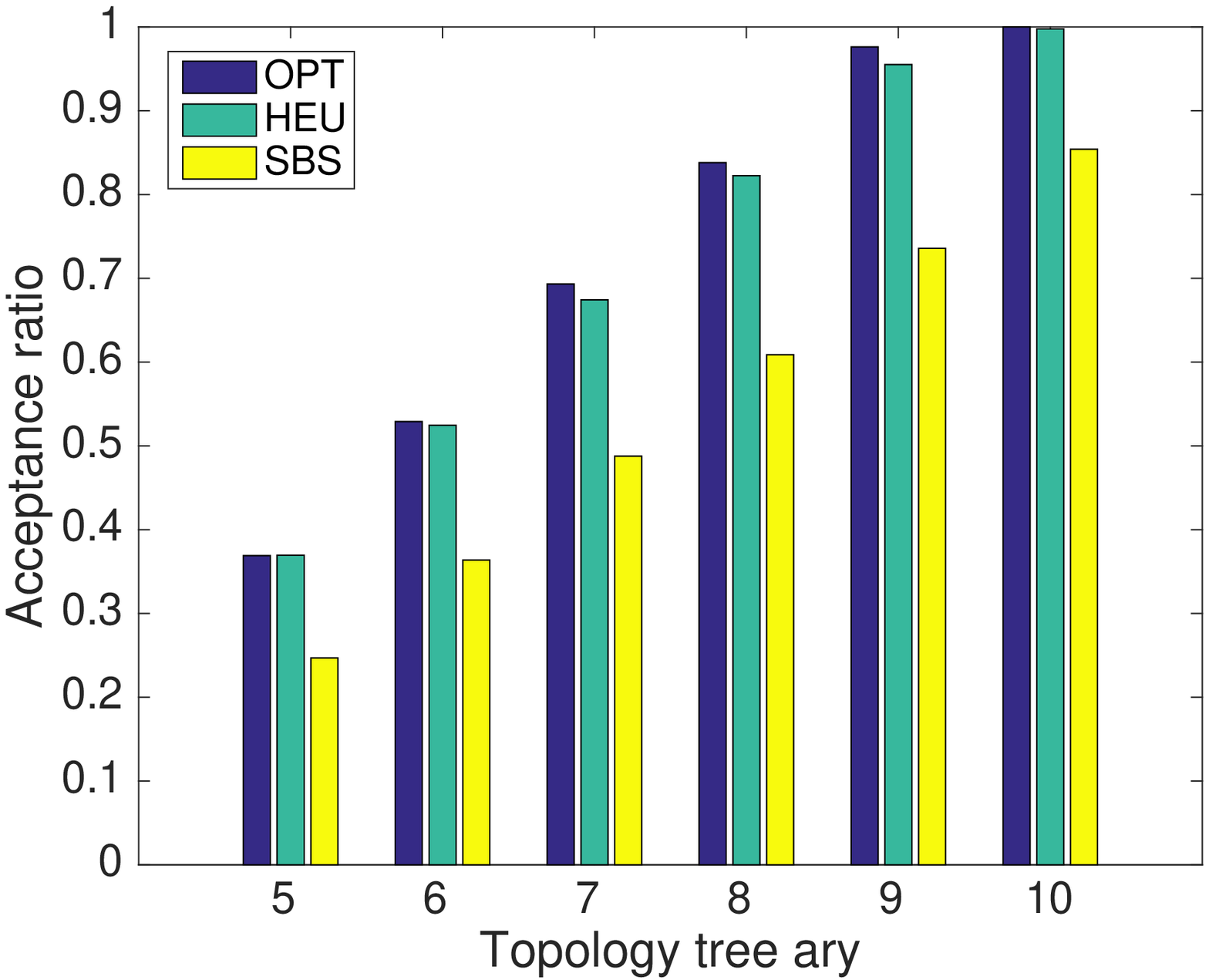}
\label{fig:ary_acc}}
\hfil
\subfloat[Average running time (log scale)]{\includegraphics[height=0.155\textwidth]{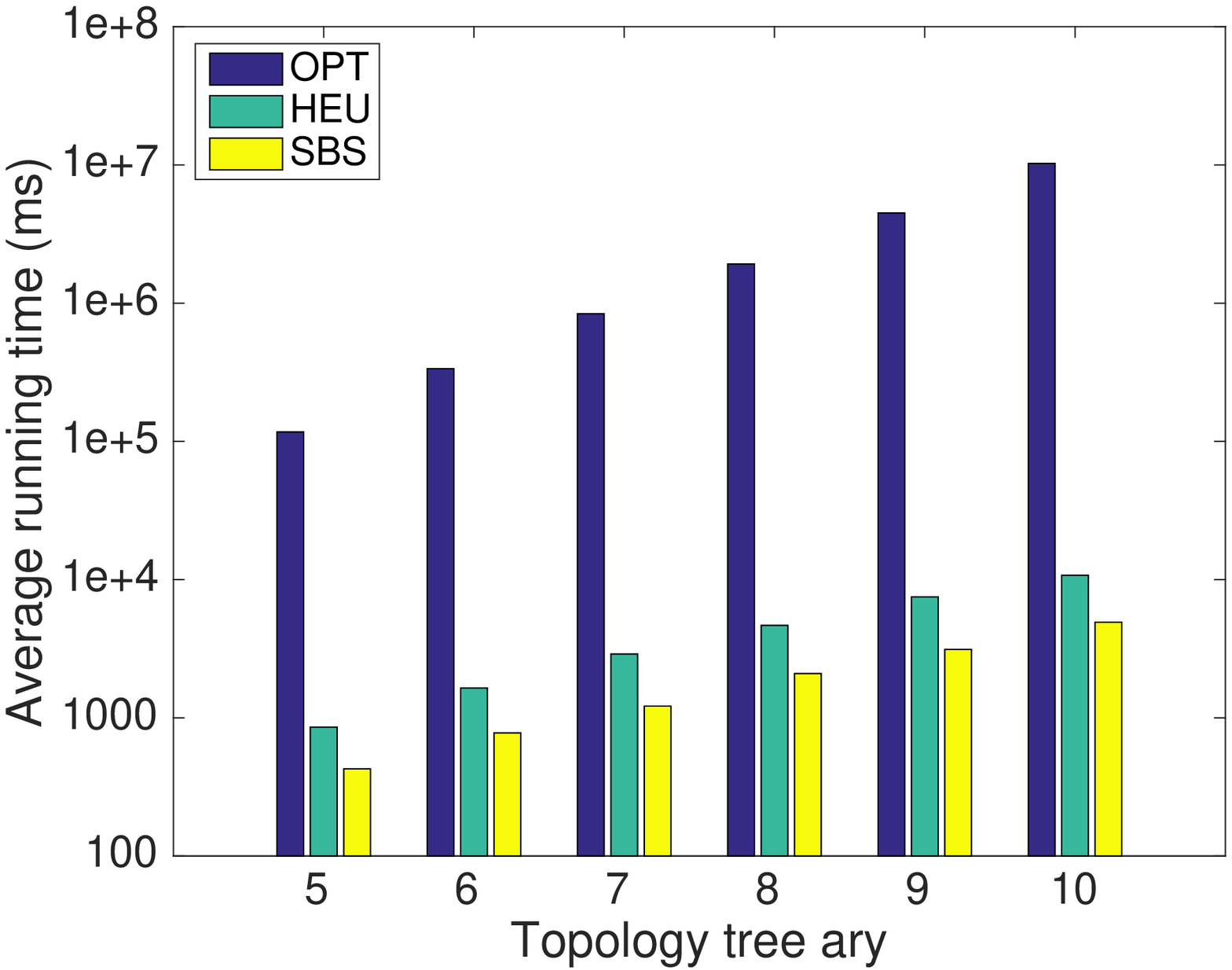}
\label{fig:ary_tm}}
\caption{Dynamic experiment with varying network size}
\label{fig:ary}
\end{figure}

In the last set of experiments, we varied the network size. %
Each topology is a $4$-level $k$-ary tree, which has $k^3$ PMs, and $k^2+k^1+1$ switches.
We varied $k$ from $5$ to $10$.
With a larger network size (and thus more VMs and bandwidth), the acceptance ratio of all algorithms increases in Fig.~\ref{fig:ary}\subref{fig:ary_acc}.
OPT and HEU both outperform SBS due to resource conservation.
The running time of all algorithms increase with the tree ary number in Fig.~\ref{fig:ary}\subref{fig:ary_tm}.
As the network itself grows linearly in the logarithmic scale, all algorithms show linear or nearly linear running time growth.

We summarize our findings as follows:
\begin{enumerate}
\item OPT guarantees per-request optimality, and has the best performance in both static and dynamic scenarios; HEU shows low acceptance ratio in the static case, but has much higher acceptance ratio in the dynamic case due to resource conservation; SBS consumes too much resources and hence performs the worst in the dynamic case.
\item Compared to OPT, HEU has much better time efficiency, which is a great advantage in practice; however, OPT is still important when 1) tenant requests are small in general, 2) cloud resources are very scarce, or 3) future researches along the same line need to compare with a theoretically (per-request) optimal solution.
\end{enumerate}
\red{
\section{Discusssions}
\label{sec:disc}

\noindent\textbf{Resource optimization:}
Our current solutions focus on minimizing the number of backup VMs.
However, they can be extended to other objectives, such as minimization of total bandwidth.
Specifically, instead of the minimum number of VMs, the minimum bandwidth to achieve a specific $\langle n_0, n_1 \rangle$ pair is computed for each node.
The aggregation process incorporates the bandwidth consumed both in lower levels and in the current level.
We omit more details due to page limit.

\noindent\textbf{Simultaneous PM failures:}
Our proposed algorithms protect from any single PM failure in the substrate.
They can be extended to cover multiple simultaneous PM failures, at the cost of exponentially increased time complexity regarding the number of failures to be covered.
Specifically, the extension involves adding $\kappa-1$ dimensions into the dynamic programming, where $\kappa$ is the number of covered simultaneous failures.
As our future work, more efficient algorithms for covering multiple simultaneous PM failures are to be developed.

\vspace{0.1in}
\begin{figure}[h!]
\centering
\includegraphics[width=0.48\textwidth]{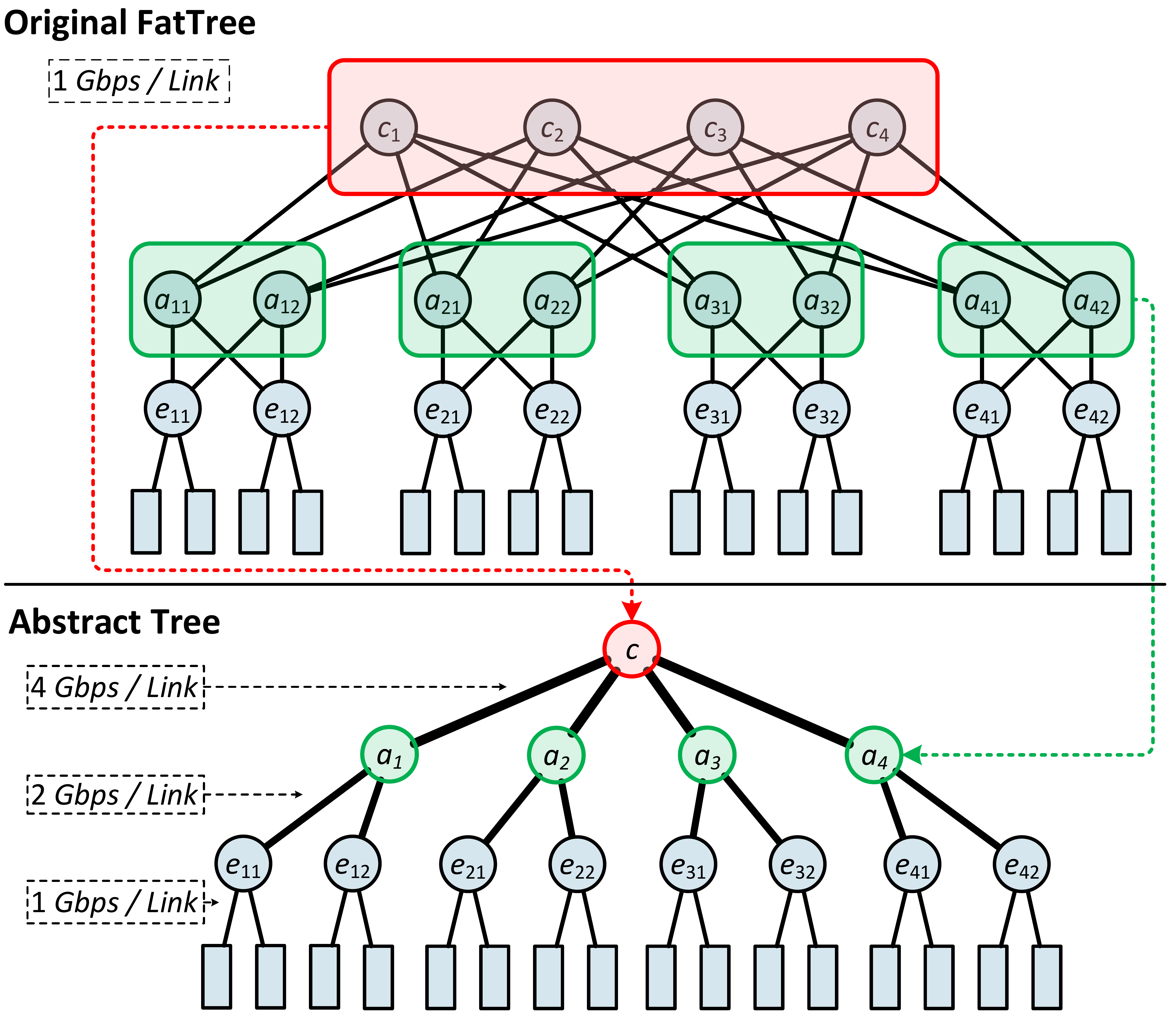}
\caption{Tree abstraction (bottom) of $4$-ary FatTree (top).}
\label{fig:fattree}
\end{figure}

\noindent\textbf{Data center topologies:}
As aforementioned, our solutions can be applied to generic tree-like topologies with simple abstractions.
To support our argument, an example is shown for the widely adopted FatTree topology~\cite{Al-Fares2008} in Fig.~\ref{fig:fattree}.
A $4$-ary FatTree topology is shown on the top, which is then abstracted as the virtual tree topology on the bottom.
Switches or links connected to the same set of lower layer nodes are aggregated into a single abstract switch or link; link capacities are also aggregated.
As the majority of data center traffic consists of small flows, we can assume arbitrary splitting of traffic between different VM pairs; hence any bandwidth allocation feasible on the aggregated topology can be successfully configured on the original topology as well.
Other topologies feasible for adopting such abstraction include VL2~\cite{Greenberg2009} and other multi-rooted tree-based topologies.
Survivable VCE for more general data center topologies is among our future directions.
}

\section{Conclusions}
\label{sec:conclusions}

\noindent In this paper, we studied survivable VC embedding with hose model bandwidth guarantee.
We formally defined the problem of minimizing VM consumption for providing survivability guarantee.
To solve the problem, we proposed a novel dynamic programming-based algorithm, with worst-case time complexity polynomial in the network size and the number of requested VMs.
We proved the optimality of our algorithm and analyzed its time complexity.
We also proposed an efficient heuristic algorithm, which is several orders faster than the optimal algorithm.
Simulation results show that both proposed algorithms can achieve much higher acceptance ratio compared to the baseline solution in the online scenario, and our heuristic algorithm can achieve similar performance as the optimal with much faster computational speed.
\bibliographystyle{IEEEtran}

\appendix
\label{sec:apdx}

\begin{proof}[Proof of Lemma~\ref{l:1}]
Preliminarily, according to Eq.~\eqref{eq:b}, a subtree offering at most $N/2$ working VMs can reduce its working VMs without increasing bandwidth on its out-bound link, and since such subtree only has child offering at most $N/2$ working VMs, such reduction does not increase link bandwidth within the subtree as well.
First, as $T_v$ can offer $n_v > N/2$ VMs, its bandwidth allocation satisfies $\mathcal{B}_s(l_v) \ge \min \{ n_v, N-n_v \} \cdot B = (N - n_v) \cdot B$.
Hence offering $n_v^- = N - n_v$ or less working VMs will not increase bandwidth on $T_v$'s out-bound link.
If $v$ is a PM, then we can directly reduce the number of working VMs on $v$ from $n_v$ to $N - n_v$.
Then, assume the lemma holds for all nodes lower than level $l$ in the tree, and let $v$ be a switch on level $l$.
If for any child $u$ of $v$, the number of working VMs that it offers $n_u \le N / 2$, then we can reduce the number of working VMs on any PM in $T_v$ without increasing bandwidth needed on any link. %
Now if some child $u$ of $v$ has $n_u > N/2$ working VMs.
Since $n_v \le N$, there is at most one such child.
First, we reduce working VMs in all other child subtrees to none.
Then, by induction from the $(l-1)$-th level, the number of working VMs on $u$ can be reduced from $n_u$ to $n_u^- = N - n_u$, which is no less than $n_v^-$, without increasing bandwidth on any link.
Now, since $n_v^- \le n_u^- \le N / 2$, we can reduce the number of working VMs  in $T_u$ without increasing bandwidth on any link, which completes the proof.
\end{proof}

\begin{proof}[Proof of Lemma~\ref{l:2}]
To start with, this always holds for any PM node $v \in H$ by simply reducing the number of VMs.
If subtree $T_v$ can offer more than $N$ working VMs, each child subtree $T_u$ of $v$ is one of the three cases: 1) $T_u$ can offer more than $N$ working VMs, 2) $T_u$ can offer working VMs in $(N/2, N]$, or 3) $T_u$ can offer no more than $N/2$ working VMs.
If there is any subtree $T_u$ in case 1, then we let all PMs not in $T_u$ to have $0$ working VMs, and continue to prove the lemma on $T_u$, until we reach any PM node.
Otherwise, if there are at least two subtrees $T_u$ and $T_w$ both in case 2, then we also let all PMs not in these two subtrees to have $0$ working VMs.
Without loss of generality, assume $T_u$ can offer $n_u$ working VMs, $T_w$ can offer $n_w$ working VMs, and $n_u \ge n_w > N/2$.
By Lemma~\ref{l:1}, $T_w$ can also offer $N - n_u \le N - n_w = n_w^-$ working VMs.
Hence we can let $T_w$ offer $N-n_u$ working VMs, and let $T_u$ offer $n_u$, the sum of which is exactly $N$.
Finally, if there is at most one subtree $T_u$ in case 2, we can always reduce the total number of working VMs not hosted in $T_u$ (or in all subtrees if $T_u$ does not exist) to $N - n_u$, without increasing bandwidth needed on any link.
This completes the proof.
\end{proof}

\begin{proof}[Proof of Theorem~\ref{th:1}]
First, we show that any feasible solution to SVCEP is feasible to SVCEP-GP with $v = r$ and $n_0 = n_1 = N$, and vice versa.
The forward direction is true based on their problem definitions.
On the reverse direction, by Lemma~\ref{l:2}, a feasible solution to SVCEP-GP with $v = r$ and $n_0 = n_1 = N$ is always able to provide exactly $N$ working VMs in any scenario.
This proves the equivalence between SVCEP and SVCEP-GP with $v = r$ and $n_0 = n_1 = N$.

We then prove by induction that SVCEP-GP has been solved optimally in Algorithm~\ref{a:1}.
For PMs, each entry is computed directly and is optimal based on definition.
For each switch node $v$, $N_v[n_0, n_1]$ is dependent on the four possible values computed in the $N_v'$ table as in Eq.~\eqref{eq:snv}.
Based on definition, if $N_v' [n_0, n_1, d_v]$ is correctly computed for every $n_0$ and $n_1$, $N_v [n_0, n_1]$ is correct in applying the bandwidth bound of $T_v$ based on Eq.~\eqref{eq:snv}.

As for $N_v' [n_0, n_1, k]$, initially when $k=0$, the values are computed directly based on definition.
Then, assume that $N_v' [n_0', n_1', k-1]$ and $N_{u_k} [n_0'', n_1'']$ are correctly computed for any $n_0', n_1', n_0'', n_1'' \in [0, N]$ respectively, where $u_k$ is the $k$-th child of node $v$.
The corresponding optimal VM allocation to the problem of $N_v' [n_0, n_1, k]$ can be split into two parts: $N^*$ in the first $(k-1)$ subtrees of $v$, and $N^{**}$ in the $k$-th subtree.
Assume that the first part can offer $n_0^*$ and $n_1^*$ VMs in the normal operation and the worst-case failure scenario respectively, and the second can offer $n_0^{**}$ and $n_1^{**}$ respectively.
Based on definition, $N^*$ and $N^{**}$ are feasible solutions to the problems $N_v' [n_0^*, n_1^*, k-1]$ and $N_{u_k} [n_0^{**}, n_1^{**}]$ respectively, and $n_0^* + n_0^{**} \ge n_0$ as well as $\min \{ n_0^* + n_1^{**}, n_0^{**} + n_1^{*} \} \ge n_1$.
\cmts{Now assume $N_v' [n_0, n_1, k]$ is not optimal.
This means that $N^* + N^{**} < N_v' [n_0, n_1, k] \le N_v' [n_0^*, n_1^*, k-1] + N_{u_k} [n_0^{**}, n_1^{**}]$, the righthand inequality due to the algorithm.
Hence we have either $N^* < N_v' [n_0^*, n_1^*, k-1]$ or $N^{**} < N_{u_k} [n_0^{**}, n_1^{**}]$, contradicting the optimality of $N_v' [n_0^*, n_1^*, k-1]$ or $N_{u_k} [n_0^{**}, n_1^{**}]$.}
This completes the proof.
\end{proof}

\begin{proof}[Proof of Theorem~\ref{th:2}]
Each node in the tree is traversed once during the DP process.
At each PM node, each entry indexed by $n_0$ and $n_1$ needs $O(1)$ time due to Eq.~\eqref{eq:hnv}, hence each PM takes $O(N^2)$ times.
For each switch node, its computation iterates over its every child $v$, and each entry indexed by $v$, $n_0$ and $n_1$ needs to iterate over up to $4$ variables, $n_0', n_1', n_0'', n_1''$, hence each switch node $v$ takes $O(d_v \cdot N^6)$ complexity.
Summing them all up, the final complexity of the DP computation process is $O(|V| \cdot N^6)$.
The time for backtracking is bounded by the DP computations.
This completes the proof, which shows that the time complexity of Algorithm~\ref{a:1} is polynomial to the network size $|V|$ and the request size $N$.
\end{proof}

\begin{proof}[Proof of Lemma~\ref{l:heu}]
We need to prove that the VCE of $\mathcal{J}'$ can provide exactly $N$ bandwidth-guaranteed VMs in any failure $F \in H$.
Let $n_v$ be the number of VMs allocated in subtree $T_v$ in the VCE; also let $\mathcal{B}(l_v)$ be the hose model bandwidth of the VCE on link $l_v$ as specified in Eq.~\eqref{eq:b}.
First, since each PM hosts at most $N'$ VMs, at least $N$ VMs remain during the failure.
Now, let $(s_0, s_1, \cdots, s_K)$ be the path from root $r$ to PM $F$, with $s_0 = r$ and $s_K = F$.
By cutting the $K$ links $(s_{k-1}, s_{k})$ for $k \le K$, the topology is partitioned into $K+1$ components, and we denote the component containing node $s_k$ as $C_k$.
$C_K$ contains only the failed PM $F$, hence is discarded.
For any link $l_v$ within each component, its allocated bandwidth satisfies
$\mathcal{B}(l_v) = \min \{ n_v, N+N'-n_v \} \cdot B \ge \min \{ n_v, N-n_v \} \cdot B$,
hence each component can offer all the allocated VMs within it for request $\mathcal{J}$, if not considering bandwidth on the cut links.
If any component $C_k$ (viewed as a subtree rooted at $s_k$) can offer at least $N$ VMs, by Lemma~\ref{l:2} it can offer exactly $N$ VMs in this scenario.
Otherwise, let $S$ be an initially empty set of VMs, we then proceed with each component $C_k$ from $k=0$ to ${K-1}$, each time adding all allocated VMs within $C_k$ to $S$, until \emph{at least} $N$ VMs are added.
Assume we stop at $C_\kappa$, where $\kappa \in [1, K-1]$.
$S$ now contains $[N, 2N-1)$ VMs, as both the first $(\kappa-1)$ components and the $\kappa$-th component contain less than $N$ VMs.
For node $s_k$ where $k \le \kappa$, let $n_{k} = N+N'-n_v$ be the number of VMs added before $C_k$, then its out-bound link $l_{s_k} = (s_{k-1}, s_{k})$ has bandwidth allocation $\mathcal{B}(l_{s_k}) = \min \{ n_k, N+N'-n_k \} \cdot B \ge \min \{ n_k, N-n_k \} \cdot B$.
Hence each link $(s_{k-1}, s_k)$ is sufficient to offer all VMs in the first $(k-1)$ components, for $k \le \kappa$.
Then we form a new tree $T'$ by pruning all components after $C_\kappa$, and assigning $s_\kappa$ as the root of $T'$.
Every link in $T'$ is either within some component $C_k$, or one of the links $(s_{k-1}, s_k)$, where $k \le \kappa$, hence its bandwidth is sufficient as above.
Now, $T'$ can offer $|S| \ge N$ working VMs, hence by Lemma~\ref{l:2} it can offer exactly $N$ working VMs.
Therefore, there is a feasible VCE for $\mathcal{J}$ in the given VCE of $\mathcal{J}'$ in any scenario,
and the lemma follows.
\end{proof}
\end{document}